\documentclass{elektr}

\usepackage{hyperref}
\hypersetup{colorlinks=true,urlcolor=blue,citecolor=blue}
\usepackage[all]{xy,xypic}
\usepackage{amsfonts,amssymb,amsmath,amsgen,amsopn,amsbsy,theorem,graphicx,epsfig}
\usepackage{eufrak,amscd,bezier,latexsym,mathrsfs,enumerate}
\usepackage[utf8]{inputenc}
\usepackage[english]{babel}
\usepackage[dvipsnames]{xcolor}
\usepackage[pagewise]{lineno}

\usepackage{wrapfig}
\usepackage{psfrag} 
\usepackage{auto-pst-pdf}
\usepackage{rotating}
\usepackage{mathptmx}
\usepackage{fancyvrb}
\usepackage{url}
\usepackage{hyperref}
\usepackage{multido}
\usepackage{listings}

\year{2019} 
\vol{}
\fpage{}
\lpage{}
\doi{}

\title{On the asymptotic analysis of the high-order statistics of the channel capacity over generalized fading channels}
\author[YILMAZ]{
\textbf{Ferkan YILMAZ\thanks{ferkan@yildiz.edu.tr}}\\
Computer Engineering Department, Faculty of Electrical \& Electronics, Y{\i}ld{\i}z Technical University, Istanbul, Turkey\\ 
ORCID iD: https://orcid.org/0000-0001-6502-8280\\ [1.8em]
\rec{.201}
\acc{.201}
\finv{..201}
}
  
\def\E{\ifmmode{\mathbb E}\else{$\mathbb E$}\fi} 
\def\N{\ifmmode{\mathbb N}\else{$\mathbb N$}\fi} 
\def\R{\ifmmode{\mathbb R}\else{$\mathbb R$}\fi} 
\def\Q{\ifmmode{\mathbb Q}\else{$\mathbb Q$}\fi} 
\def\C{\ifmmode{\mathbb C}\else{$\mathbb C$}\fi} 
\def\H{\ifmmode{\mathbb H}\else{$\mathbb H$}\fi} 
\def\Z{\ifmmode{\mathbb Z}\else{$\mathbb Z$}\fi} 
\def\P{\ifmmode{\mathbb P}\else{$\mathbb P$}\fi} 
\def\T{\ifmmode{\mathbb T}\else{$\mathbb T$}\fi} 
\def\SS{\ifmmode{\mathbb S}\else{$\mathbb S$}\fi} 
\def\DD{\ifmmode{\mathbb D}\else{$\mathbb D$}\fi} 

\newcommand{\bse}{\begin{subequations}}
\newcommand{\ese}{\end{subequations}}
\newcommand{\ben}{\begin{enumerate}}
\newcommand{\een}{\end{enumerate}}
\newcommand{\bens}{\begin{enumerate*}}
\newcommand{\eens}{\end{enumerate*}}
\newcommand{\be}{\begin{equation}}
\newcommand{\ee}{\end{equation}}
\newcommand{\bea}{\begin{eqnarray}}
\newcommand{\eea}{\end{eqnarray}}
\newcommand{\baa}{\begin{eqnarray*}}
\newcommand{\eaa}{\end{eqnarray*}}
\newcommand{\bc}{\begin{center}}
\newcommand{\ec}{\end{center}}

\newtheorem{theorem}{Theorem}

\theoremstyle{corollary}

\theoremstyle{lemma}
\newtheorem{lemma}{Lemma}

\theoremstyle{proposition}

\theoremstyle{axiom}

\theoremstyle{conjecture}

\theoremstyle{example}

\theoremstyle{definition}

\theoremstyle{remark}

\newcommand{\theoremref}[1]{Theorem~\protect\ref{#1}}
\newcommand{\lemmaref}[1]{Lemma~\protect\ref{#1}}

\newcommand{\CHECKQEDsymbol}{$\hfill\checkmark$}

\DefineVerbatimEnvironment{code}{Verbatim}{fontsize=\small}
\newcommand{\figref}[1]{Figure~\protect\ref{#1}}

\newcommand{\secref}[1]{Section~\protect\ref{#1}}

\newcommand\trigeq{\mathrel{\overset{\makebox[0pt]{\mbox{\normalfont\tiny\sffamily$\!\scriptscriptstyle\triangle\!$}}}{=}}}
 
\def\suscript(#1,#2,#3){{#1}^{#2}_{#3}}

\newcommand{\msub}[1]{{\suscript(m,{},{#1})}}

\newcommand{\betasub}[1]{{\suscript(\beta,{},{#1})}}

\newcommand{\gammasub}[1]{{\suscript(\gamma,{},{#1})}}
\newcommand{\gammabarsub}[1]{{\suscript(\bar{\gamma},{},{#1})}}
\newcommand{\xisub}[1]{{\suscript(\xi,{},{#1})}}

\newcommand{\gammabar}{{\bar{\gamma}}}

\newcommand\scalemath[3]{\scalebox{#2}[#1]{\mbox{\ensuremath{\displaystyle{#3}}}}}
\newcommand{\fracparams}[2]{\genfrac{}{}{0pt}{}{{#1}}{{#2}}}
\newcommand{\FoxYDefinition}[3]{\suscript(\rm{Y},{#1},{#2}){\left[#3\right]}}
\let\prevequation\equation\def\equation{\setlength\abovedisplayskip{3.25pt}\setlength\belowdisplayskip{3.25pt}\prevequation}
\newcommand{\FoxY}[6][right]{
	\ifthenelse{\equal{#1}{right}}{\suscript(\rm{Y},{#2},{#3}){\left[{#4}\left|\fracparams{#5}{#6}\right.\right]}}{
		\ifthenelse{\equal{#1}{left}}{\suscript(\rm{Y},{#2},{#3}){\left[\left.{#4}\right|\fracparams{#5}{#6}\right]}}{
			\suscript(\rm{Y},{#2},{#3}){\left[{#4}\left|\fracparams{#5}{#6}\right.\right]}
		}
	}
}

\let\preveqnarray\eqnarray\def\eqnarray{\setlength\abovedisplayskip{3.25pt}\setlength\belowdisplayskip{3.25pt}\setlength\arraycolsep{1.4pt}\preveqnarray}
\newcommand{\FoxH}[6][right]{
	\ifthenelse{\equal{#1}{right}}{\suscript(\rm{H},{#2},{#3}){\left[{#4}\left|\fracparams{#5}{#6}\right.\right]}}{
		\ifthenelse{\equal{#1}{left}}{\suscript(\rm{H},{#2},{#3}){\left[\left.{#4}\right|\fracparams{#5}{#6}\right]}}{
			\suscript(\rm{H},{#2},{#3}){\left[{#4}\left|\fracparams{#5}{#6}\right.\right]}
		}
	}
}
\newcommand{\MeijerGDefinition}[3]{\suscript(\rm{G},{#1},{#2}){\left[#3\right]}}
\let\prevmultline\multline\def\multline{\setlength\abovedisplayskip{3.25pt}\setlength\belowdisplayskip{3.25pt}\prevmultline}
\newcommand{\MeijerG}[6][right]{
	\ifthenelse{\equal{#1}{right}}{\suscript(\rm{G},{#2},{#3}){\left[{#4}\left|\fracparams{#5}{#6}\right.\right]}}{
		\ifthenelse{\equal{#1}{left}}{\suscript(\rm{G},{#2},{#3}){\left[\left.{#4}\right|\fracparams{#5}{#6}\right]}}{
			\suscript(\rm{G},{#2},{#3}){\left[{#4}\left|\fracparams{#5}{#6}\right.\right]}
		}
	}
}
\newcommand{\Hypergeom}[6]{
	\suscript({},{},{#1})\suscript({F},{#3},{#2})\!\left[{#4};{#5};{#6} \right]
}
\newcommand{\DHypergeom}[9]{
	{\suscript({},{},{#1})\suscript({\mathbb{F}},{#3},{#2})\!\left[\fracparams{#4}{#7}\!\fracparams{;}{;}\!\fracparams{#5}{#8}\!\fracparams{;}{;}\!\fracparams{#6}{#9}\right]}
	}
\newcommand{\DHypergeomSMALL}[9]{
	{\suscript({},{},{#1})\suscript({\mathbb{F}},{#3},{#2})\!\bigl[\textstyle\fracparams{#4}{#7}\!\fracparams{;}{;}\!\fracparams{#5}{#8}\!\fracparams{;}{;}\!\fracparams{#6}{#9}\bigr]}
	}
	
\newcommand{\BesselI}[2][0]{
	\suscript({I},{},{#1})\left({#2}\right)
}

\newcommand{\HeavisideTheta}[1]{
	{\theta}\left({#1}\right)
}
\newcommand{\DiracDelta}[1]{
	{\delta\left({#1}\right)}
}
\newcommand{\KroneckerDelta}[2]{
	{\delta_{{#1},{#2}}}
}

\newcommand{\Binomial}[2]{
	\binom{#1}{#2} 
}
\newcommand{\Multinomial}[2]{
	\binom{#1}{#2} 
}

\newcommand{\ExtGamma}[4]{
	{\rm{\Gamma}}\left({#1},{#2},{#3},{#4}\right)
}

\newcommand{\PolyGamma}[2]{
	\ifthenelse{\equal{#1}{0}}{{\psi}\left({#2}\right)}
		{{\psi}^{#1}\left({#2}\right)}
}

\newcommand{\Expected}[1]{
	{\,{\mathbb{E}}\!\left[{#1}\right]}
}
\newcommand{\argmin}{\arg\!\min}
\newcommand{\imaginary}{{\rm{i}}}
\newcommand{\EulerGamma}{{\boldsymbol{\rm{E}}}}

\DeclareMathOperator\erf{erf}
\DeclareMathOperator\inverf{inverf}
\newcommand{\FourierTransform}[4][norm]{
	\ifthenelse{\equal{#1}{norm}}{
	{\mathcal{F}_{#2}\left\{{#3}\right\}\left({#4}\right)}
	}{
		\ifthenelse{\equal{#1}{conj}}{
		{\mathcal{F}_{#2}^{*}\left\{{#3}\right\}\left({#4}\right)}
		}{
		{\mathcal{F}_{#2}\left\{{#3}\right\}\left({#4}\right)}
		}
	}
}
\newcommand{\phantomas}[3][c]{\ifmmode\makebox[\widthof{$#2$}][#1]{$#3$}\else\makebox[\widthof{#2}][#1]{#3}\fi}
\newlength{\widebarargwidth}
\newlength{\widebarwidth}
\newlength{\widebarargheight}
\newlength{\widebarargdepth}
\DeclareRobustCommand{\widebar}[1]{%
	\settowidth{\widebarargwidth}{\ensuremath{#1}}%
	\settoheight{\widebarargheight}{\ensuremath{#1}}%
	\settodepth{\widebarargdepth}{\ensuremath{#1}}%
	\addtolength{\widebarargwidth}{-0.2\widebarargheight}%
	\addtolength{\widebarargwidth}{-0.2\widebarargdepth}%
	\makebox[0pt][l]{\addtolength{\widebarargheight}{0.3ex}%
		\hspace{0.2\widebarargheight}%
		\hspace{0.2\widebarargdepth}%
		\hspace{0.5\widebarargwidth}%
		\setlength{\widebarwidth}{0.61803\widebarargwidth}%
		\addtolength{\widebarwidth}{0.3ex}%
		\makebox[0pt][c]{\rule[\widebarargheight]{\widebarwidth}{0.1ex}}}{#1}}

\newcommand{\mathsym}[1]{{}}

\newcommand{\fact}[1]{{#1}!}

\newcommand{\inverseF}{{\rm{inv}}}
\newcommand{\Mathematica}{{\scshape{Mathematica}}\textsuperscript{\tiny\textregistered}\,}
\newcommand{\Matlab}{{\scshape{Matlab}}\textsuperscript{\tiny{TM}}\,}
\newcommand{\Maple}{{\scshape{Maple}}\textsuperscript{\tiny\textregistered}\,}

\begin{document}

\maketitle

\begin{abstract}
In this article, we provide further asymptotic analysis to the \emph{higher-order statistics} (HOS) of the channel capacity over generalized fading channels, especially by proposing simple and closed-form expressions each of
which can be easily computed as a tight-bound revealing the existence of constant gap between the actual and asymptotic HOS of of the channel capacity in the limit of both high- and low-\emph{signal to noise ratios} (SNRs).
As such, we show that these closed-form asymptotic expressions are insightful enough to comprehend the diversity gains. The mathematical formalism we followed in this article is illustrated with some selected numerical examples that 
validate the correctness of our newly derived asymptotic results.
\end{abstract}

\section{Introduction}\label{Sec:SectionI}
In literature of wireless communications, the \emph{first-order statistics} (FOS) of the \emph{channel capacity} (CC) is well-known as \emph{averaged} CC (ACC) and defined by $\widebar{C}(\gammabar_{end})=\mathbb{E}[\log(1+\gamma_{end})]$ for a certain averaged SNR $\gammabar_{end}=\mathbb{E}[\gamma_{end}]$, where $\gamma_{end}$ denotes the end-to-end instantaneous \emph{signal-to-noise ratio} (SNR), $\mathbb{E}[\cdot]$ denotes the expectation operator, and $\log(\cdot)$ denotes the natural logarithm. We note that the FOS of the channel capacity has been extensively investigated in literature, considering different fading environments \cite[and references therein]{BibYilmazAlouiniUnifiedPerformanceTCOM2012,BibYilmazAlouiniTCOM2012,BibYilmazAlouiniPIMRC2010,BibYilmazAlouiniICC2012,BibKhairiAshourHamdi2008,BibYilmazAlouiniISWCS2010,BibYilmazAlouiniEGK2010,BibDiRenzoGraziosiSantucci2010}. In particular, Yilmaz \& Alouini proposed in \cite{BibYilmazAlouiniICC2012} a \emph{moment generating function} (MGF)-based approach for the ACC analysis, specifically introducing how to unify the ACC analyses of diversity combining and transmission schemes into a single MGF-based analysis. Recently, in order to ensure the reliability of wireless transmissions in addition to its quality, many theoreticians, practitioners and researchers\cite{BibDiRenzoGraziosiSantucci2010,BibSagiasWPC2011,BibYilmazAlouiniWCLetter2012,BibYilmazTabassumAlouiniMELECON2012,BibYilmazTabassumAlouiniTVT2014,BibPeppasMathiopoulosZhangSasase2018,BibTsiftsisFoukalasKaragiannidisKhattabTVT2016,BibZhangChenPeppasLiLiuTCOM2017} turn their attention to the \emph{higher-order channel capacity} (HOCC), defined as the \emph{higher-order statistics} (HOS) of the CC, that is $\widebar{C}(n;\gammabarsub{end})=\mathbb{E}[\log^n(1+\gamma_{end})]$, where $n\in\mathbb{N}$ denotes the order of the statistics. The HOCC is easily utilized  to statistically characterize the maximum transmission throughput over fading channels with a negligible small \emph{bit error rate} (BER). However, since being analytically intractable, the HOCC over fading channels has scarce literature, especially when compared to the FOS. In particular, to achieve the exact HOCC analysis, Di Renzo \emph{et. al} proposed in \cite[Theorem 6]{BibDiRenzoGraziosiSantucci2010} a numerical approach for \emph{maximum ratio combining} (MRC) over generalized fading environments. However, we note that this numerical approach cannot be easily performed in closed-form even in typical scenarios of wireless transmission. The next framework \cite{BibSagiasWPC2011}, presented by Sagias \emph{et al.}, is a \emph{probability density function} (PDF)-based framework valid only for diversity combining receivers in Rayleigh and Nakagami-\emph{m} fading environments. Later, presented in \cite{BibYilmazAlouiniWCLetter2012} is the first MGF-based approach for the exact HOCC anaysis in fading environments such that it eliminates all difficulties emerged from \cite{BibDiRenzoGraziosiSantucci2010} and \cite{BibSagiasWPC2011}. The HOCC analysis is also studied in \cite{BibYilmazTabassumAlouiniMELECON2012,BibYilmazTabassumAlouiniTVT2014} for \emph{multihop transmission} (MT), in \cite{BibPeppasMathiopoulosZhangSasase2018} for \emph{equal gain combining} (EGC), in \cite{BibTsiftsisFoukalasKaragiannidisKhattabTVT2016} for dispersed spectrum cognitive radio over $\eta  -\mu$ fading channels, and in \cite{BibZhangChenPeppasLiLiuTCOM2017} for spectrum aggregation systems. However, all HOCC analyses, mentioned above, generally lead to complicated expressions \emph{either involving a single integration or the evaluation of advanced special functions} ---  they are therefore not insightful enough to comprehend the diversity gains.
In this context, Yilmaz and Alouini introduced in \cite{BibYilmazSPAWC2012} a simple and comprehensive asymptotic framework for the analysis of the HOCC over generalized composite fading environments.

In agreement with the literature above, let $\bar{C}(n;\gammabar_{end})$ denote the the HOCC in generalized fading environment, and $\bar{C}_{\text{AWGN}}(n;\gammabar_{end})$ denote the HOCC of \emph{additive white Gaussian noise} (AWGN) channel in non-fading environment. We note that there exist gaps between $\bar{C}(n;\gammabar_{end})$ and $\bar{C}_{\text{AWGN}}(n;\gammabar_{end})$ both in high-SNR or low-SNR regimes since these two environment have distinct diversity gains. Accordingly, we show that these gaps are constant in log-domain, as seen in \figref{Fig:ConstantCapacityGapBetweenDifferentFadingEnvironments}, that is  
\begin{equation}\label{Eq:ConstantGaps}
\lim_{\gammabar_{end}\rightarrow\infty}\log\left(\frac{\bar{C}_{\text{AWGN}}(n;\gammabar_{end})}{\bar{C}(n;\gammabar_{end})}\right)\approx\Delta_{\text{High-SNR}},
\text{ and }
\lim_{\gammabar_{end}\rightarrow{0}}\log\left(\frac{\bar{C}(n;\gammabar_{end})}{\bar{C}_{\text{AWGN}}(n;\gammabar_{end})}\right)\approx\Delta_{\text{Low-SNR}},
\end{equation}
where $\Delta_{\text{High-SNR}}\!\geq\!0^{+}$ and $\Delta_{\text{Low-SNR}}\!\geq\!0^{+}$ are the positive constant gaps in high- and low-SNR regimes, respectively. Further, we propose novel expressions on the asymptotic HOCC analysis \emph{in terms of computable closed-form expressions}, each of which is insightful enough to comprehend the diversity gains. In addition to our asymptotic and closed-form expressions, we obtain the boundary SNR values specifying for which SNR values the wireless transmission in a specified fading environment certainly enters either in high-SNR or low-SNR regimes. In this context, we exemplified our
\begin{wrapfigure}{r}{0.5\textwidth}
  \begin{center}
  	\vspace{-9mm}
	\psfrag{XLabel}[c]{\footnotesize SNR [dB]}
	\psfrag{YLabel}[c]{\footnotesize Higher-Order Ergodic Capacity}
	\includegraphics[width=0.5\textwidth,keepaspectratio=true]{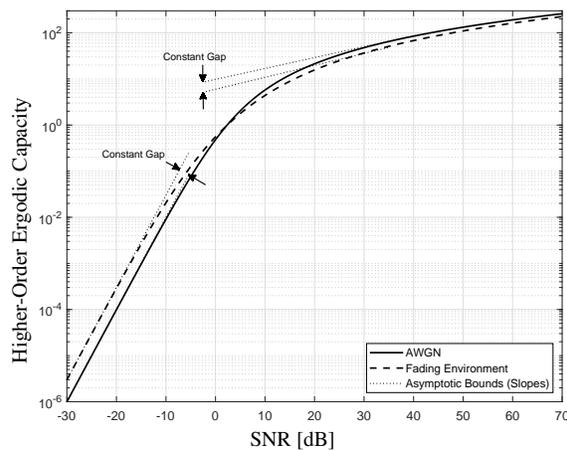}
	\vspace{-6mm}
	\caption{Example illustration of the existence of constant gap between the HOCCs of different fading environments.}
	\label{Fig:ConstantCapacityGapBetweenDifferentFadingEnvironments}
	\vspace{-5mm}
  \end{center}
  \vspace{-14mm}
\end{wrapfigure} 
asymptotic and closed-form expressions for different generalized fading environments commonly used in the literature. All these expressions have been checked numerically for their validation, correctness and accuracy.

We organized the remainder of this article as follows. In \secref{Sec:SectionII}, we shortly explain the notations used through the article, and establish the channel and system model, and discuss in details the HOS of the CC and the
HOCC. In \secref{Sec:SectionIII} and \secref{Sec:SectionIV}, we provide simple asymptotic and closed-form expressions for the HOCC analysis both in high-SNR and low-SNR regimes, respectively, with the introduction of the auxiliary coefficient needed for the asymptotic analysis. Eventually, some key results are presented for some well-known fading environments. Moreover, in \secref{Sec:SectionIII} and \secref{Sec:SectionIV}, we define boundary SNR values to determine for which SNR values high- or low-SNR regime starts. With the aid of the boundary SNR values, we extends our asymptotic to analysis to the analysis of the higher-order ergodic capacity in low-SNR regimes. These newly obtained results are presented and applied to some well-known fading environments. Finally, we summarize the main results and draw some conclusions in the last section.

\section{HOCC in generalized fading environments}\label{Sec:SectionII}
Let us consider a transmission of band-limited signal over AWGN channel in fading environments. We can write the mathematical model of the signal received by the receiver as follows 
\begin{equation}\label{Eq:SingleLinkModel}
	R_{end}=\alpha{S}+N,
\end{equation}
where $R_{end}$ denotes the received signal, $S$ denotes the transmitted symbol with the average power $E_{S}\trigeq\mathbb{E}[|S|^2]$, and $\alpha$ denotes the fading amplitude. Further, $N$ denotes the complex AWGN with zero mean and $N_{0}/2$ variance per dimension. Without loss of generality, we assume that the \emph{channel state information} (CSI) is known at the receiver. Coherent detection at the receiver is therefore possible, which results that the instantaneous SNR $\gamma_{end}$ during one symbol duration is given by $\gamma_{end}=\frac{E_s}{N_0}\alpha^2$, where $\alpha$ is evidently considered as a random variable over all transmission. The statistical moments of $\gamma_{end}$, i.e., $\Expected{\gamma^{n}_{end}}$ for all $n\in\mathbb{N}$ are therefore important for statistical characterization. 

We introduce in \cite{BibYilmazSPAWC2012} the $n$th order \emph{amount of
fading} (AOF), that is
\begin{equation}\label{Eq:HigherOrderAmountOfFading}
	AF_{n}(\gammabar_{end})=\frac{\Expected{\gamma^{n}_{end}}}{\Expected{\gamma_{end}}^{n}}-1.
\end{equation}
where $n\in\mathbb{N}$ denotes the order of the statistics. The first order AOF, often called AOF in the literature, is first introduced in \cite{BibCharash1979} as a measure of fading severity. In more details, the inverse of the first order AOF $AF_{1}(\gammabar_{end})$ is a good approximation for the number of mutually independent replicas of the transmitted signals, that is $m_\text{order}\approx{1}/{AF_{1}(\gammabar_{end})}$. Since the higher-order AOF $AF_{n}(\gammabar_{end})$ is characterized by the fading parameters, rather than by the average power $\gammabarsub{end}$, it is successfully utilized in \cite{BibYilmazAlouiniIWCMC2009} as a tool to statistically characterize $\gamma_{end}$ for \emph{moment generating function} (MGF), PDF, and \emph{cumulative distribution function} (CDF), especially by exploiting $AF_{n}(\gammabar_{end})$ in form of Laguerre moments (i.e., see \cite[Eqs. (24), (25), and (26)]{BibYilmazAlouiniIWCMC2009}). 

Further, following in the same manner, we introduce the other performance metric, known as the instantaneous HOCC $C\left(n;\gamma_{end}\right)$, whose definition is given by \begin{equation}\label{Eq:InstantaneousHigherOrderChannelCapacity}
	C\left(n;\gamma_{end}\right)=\log^{n}\left(1+\gamma_{end}\right),\quad\text{\emph{nats\textsuperscript{$n$}/\,sn\textsuperscript{$n$}/\,Hz\textsuperscript{$n$}}},
\end{equation}
where $n\in\mathbb{N}$ denotes the order of the statistics. Due to reflections, refractions, scattering and multipath propagation in fading environments, it is evidently considered as a random variable whose average $\widebar{C}\left(n;\gammabar_{end}\right)=\mathbb{E}[C\left(n;\gamma_{end}\right)]$ is known either as the HOS of the CC or the HOCC in the literature, and readily given by 
\begin{equation}\label{Eq:HOCC}
	\widebar{C}(n;\gammabarsub{end})=\int_{0}^{\infty}\log^{n}\left(1+\gamma\right)\,p_{\gamma_{end}}\left(\gamma;\gammabarsub{end}\right)d\gamma,
		\quad\text{\emph{nats\textsuperscript{$n$}/\,sn\textsuperscript{$n$}/\,Hz\textsuperscript{$n$}}},
\end{equation}
where $p_{\gamma_{end}}\left(\gamma;\gammabarsub{end}\right)$ denotes the PDF of $\gamma_{end}$. Regarding the efficiency in  numerical computation, we can say that it is tedious and complicated to obtain \eqref{Eq:HOCC} in simple and
easy-to-compute closed-form expressions, essentially because of the order $n\in\mathbb{N}$. In this context, there exist two regimes for the instantaneous SNR $\gamma_{end}$, the one of which is the high-SNR regime concerned with the diversity-multiplexing trade-off \cite{BibZhengTse2003}. The other one is the low-SNR regime involved with the energy efficiency \cite{BibVerduTIT2002,BibZhengMedardTIT2007,BibLiJinMcKayMatthewGaoWongTCOM2010,BibLozanoTulinoVerduTIT2003}. The complexity of the numerical analysis for the HOCC becomes more of an issue for high-SNR and low-SNR regimes. In the following sections, we introduce a general framework yielding asymptotically tight-bounds to achieve the analysis of the HOCC $\widebar{C}(n;\gammabarsub{end})$, and also attaining that of the ACC as a consequence.

\section{HOCC in high-SNR regime of generalized fading environments}\label{Sec:SectionIII}
In this section, we deal with a sharp characterization of the HOCC in the form of an asymptotically tight lower-bound in high-SNR $\gammabarsub{end}\gg{1}$ regime. 

\begin{theorem}\label{Theorem:FadingEnvironmentInHighSNRRegime}
Any fading environment leads in high-SNR regime when    
\begin{equation}\label{Eq:LimitSNRValueForHighSNRRegime}
	\gammabar_{end}\gtrapprox{17.5848747065}~~({12.4513927829}~\text{\rm{dB}}).
\end{equation}
\end{theorem}

\begin{proof}
Note that the CDF of $\gammasub{end}$ is defined by $P_{\gammasub{end}}(\gamma;\gammabarsub{end})=\mathbb{E}[\HeavisideTheta{\gamma-\gammasub{end}}]$, where $\HeavisideTheta{\cdot}$ denotes the Heaviside's theta function \cite[Eq.(1.8.3)]{BibZwillingerBook}. Particularly, in the high-SNR regime, we have $P_{\gammasub{end}}(\gamma_{th};\gammabarsub{end})\approx{0}$ for a certain specified threshold $\gamma_{th}\ll\gammabar_{end}$, thus  we reduce
$\widebar{C}(n;\gammabarsub{end})$ to an asymptotically tight lower bound, that is
\begin{equation}\label{Eq:ApproximatedHOCC}
\widebar{C}(n;\gammabarsub{end})
	\,\gtrapprox\,\int_{\gamma_{th}}^{\infty}\log^{n}\left(1+\gamma\right)\,p_{\gamma_{end}}\left(\gamma;\gammabarsub{end}\right)d\gamma
	\,\ge\,\int_{\gamma_{th}}^{\infty}\log^{n}\left(\gamma\right)\,p_{\gamma_{end}}\left(\gamma;\gammabarsub{end}\right)d\gamma.
\end{equation}
where we suggest to choose $\gamma_{th}\geq{8}$ by numerical observation, especially based on $\log^{n}\left(1+\gamma_{th}\right)\approx\log^{n}(\gamma_{th})$ whose error is given by $\epsilon=1-{{\log^{n}(\gamma_{th})}/{\log^{n}(1+\gamma_{th})}}$, where $\epsilon\in\mathbb{R}^{+}$ is a small number. By saying asymptotically tight-bound, we mean that the difference between the actual HOCC and its lower bound 
becomes zero in the high-SNR; and accordingly we readily note that the high-SNR regime certainly starts when
\begin{equation}\label{Eq:HighSNRRegimeMeasure}
	\psi(\gamma_{th};\gammabarsub{end})=\frac{P_{\gammasub{end}}(\gamma_{th};\gammabarsub{end})}{1-P_{\gammasub{end}}(\gamma_{th};\gammabarsub{end})}\leq{1}.
\end{equation}
Frankly speaking, the fading environment leads in high-SNR regime when $\psi(\gamma_{th};\gammabarsub{end})=1/2$, specially meaning that $\gamma_{end}$ falls below or exceeds $\gamma_{th}$ with the half probability (i.e., $P_{\gammasub{end}}(\gamma_{th};\gammabarsub{end})=1/2$). Accordingly, when $\psi(\gamma_{th};\gammabarsub{end})\ll{1}$, the fading environment is certainly in high-SNR regime. Therefore, for the critical
value $P_{\gammasub{end}}(\gamma_{th};\gammabarsub{end})=1/2$, the minimum average SNR at which transmission leads in the high-SNR regime is given by
\begin{equation}\label{Eq:SNRValueForHighSNRRegime}
	\gammabarsub{end}=\inverseF_{\gammabar}\{P_{\gammasub{end}}(\gamma_{th};\gammabar)\}({1}/{2}),
\end{equation}
where $\inverseF_{x}\{f(x)\}(y)$ denotes the inverse function of $y=f(x)$ such that $\inverseF_{x}\{f(x)\}\left(f(x)\right)=x$. We emphasize that that in the course of linear or non-linear fading environments, the worst-case fading amplitudes are well-known characterized by a one-sided Gaussian fading \cite[Section~2.2.1]{BibAlouiniBook}, so the worst-case SNR distribution follows the CDF given by $P_{\gammasub{end}}(\gamma;\gammabarsub{end})=\erf(\sqrt{0.5{\gamma}/{\gammabarsub{end}}})$, where $\erf(\cdot)$ denotes the error function \cite[Eq. (8.250/1)]{BibGradshteynRyzhikBook}. Consequently, applying \eqref{Eq:SNRValueForHighSNRRegime} on the CDF, we can readily derive $\gammabar_{end}=0.5{\gamma_{th}}/{\inverf^{2}(0.5)}$, where $\inverf(\cdot)$ denotes the inverse error function \cite{BibAbramowitzStegunBook}. Choosing $\gamma_{th}=8$ for a good approximation, $\gammabar_{end}$ will be obtained as in \eqref{Eq:LimitSNRValueForHighSNRRegime}, which proves \theoremref{Theorem:FadingEnvironmentInHighSNRRegime}
\end{proof}

Note that, with the aid of both Jensen's inequality (i.e., $\log^{n}(\mathbb{E}[\gammasub{end}])\leq\mathbb{E}[\log^{n}(\gammasub{end})]$) \cite[Sec.~12.411]{BibGradshteynRyzhikBook} and the fact that $\log(\cdot)$ is a monotonically increasing function over $\mathbb{R}^{+}$,  we can reduce $\widebar{C}(n;\gammabarsub{end})$ to an asymptotically tight lower bound which is well-known in the literature, that is 
\begin{equation}\label{Eq:JensenEqualityBasedApproximatedHOECapacity}
    \widebar{C}(n;\gammabarsub{end})\gtrapprox\log^{n}(\gammabarsub{end}),
\end{equation}
which is however not insightful enough to comprehend the diversity gains since not revealing the existence of differences (i.e., constant gaps in log-domain) among the HOCCs of different fading environments. On the other hand, an asymptotically tight lower bound of the HOCC in generalized fading environments is obtained in the following theorem.

\begin{theorem}\label{Theorem:HighOrderErgodicCapacityForHighSNRRegime}
An asymptotically tight lower-bound of $\widebar{C}(n;\gammabarsub{end})$ in high-SNR regime of fading environments is given by  
\begin{equation}\label{Eq:HighOrderErgodicCapacityForHighSNRRegime}
\widebar{C}(n;\gammabarsub{end})\geq
	\log^{n}\left(\gammabarsub{end}\right)+
	\sum_{k=0}^{n}
	\Binomial{n}{k}
	\mu_{k}\,\log^{n-k}\left(\gammabarsub{end}\right),
	\text{~~where~~}
	\mu_{k}=\Bigl.\frac{\partial^{k}}{\partial{n}^{k}}AF_{n}(\gammabar_{end})\Bigr|_{n=0}.
\end{equation}
\end{theorem} 

\begin{proof}
Using \cite[Eqs. (5.1.4/a) and (5.1.4/i)]{BibZwillingerBook}, we express $\log^{n}(\gammasub{end})$ in terms of the $n$th order differentiation, that is
\begin{equation}
 \log^{n}(\gamma_{end})=\gamma_{end}^{-k}\,\Bigl\{
	\alpha^{k}\log^{n}(\kappa)+
		\Bigl(\!\frac{\partial}{\partial{k}}\!\Bigr)^{\!n}\!\bigl(\gamma_{end}^{k}-\kappa^{k}\bigr)\Bigr\},   
\end{equation}
where $n,k\in\mathbb{N}$ and the constant term $\kappa\in\mathbb{C}$. The right part of this result depends on a arbitrary value of $k$, but neither does the left part. For simplicity, we could choose $k=0$. Then, we have $\log^{n}(\gamma_{end})=\log^{n}(\kappa)+\lim_{k\rightarrow{0}}\,\bigl(\!\frac{\partial}{\partial{k}}\!\bigr)^{\!n}\!\bigl(\gamma_{end}^{k}-\kappa^{k}\bigr)$. Subsequently, substituting this result into \eqref{Eq:HOCC} and exploiting \eqref{Eq:ApproximatedHOCC}, we simply obtain $\widebar{C}(n;\gammabarsub{end})$ as follows
\begin{equation}
\label{Eq:AFApproximatedHOCC}
\widebar{C}(n;\gammabarsub{end})\gtrapprox\int_{0}^{\infty}\log^{n}\left(\gamma\right)\,p_{\gamma_{end}}\left(\gamma;\gammabarsub{end}\right)d\gamma 
	=\log^{n}(\kappa)+\lim_{k\rightarrow{0}}\left(\frac{\partial}{\partial{k}}\right)^{\!n}\int_{0}^{\infty}(\gamma_{end}^{k}-\kappa^{k})\,p_{\gamma_{end}}\left(\gamma;\gammabarsub{end}\right)d\gamma.
\end{equation}
Further, using the definition of the $k$th moment of $\gamma_{end}$, i.e., $\mathbb{E}[\gamma_{end}^{k}]\trigeq\int_{0}^{\infty}\gamma^{k}p_{\gammasub{end}}(\gamma)d\gamma$, and then choosing $\kappa=\gammabar_{end}$, \eqref{Eq:AFApproximatedHOCC} is easily rewritten as $\widebar{C}(n;\gammabarsub{end})\gtrapprox\log^{n}(\gammabar_{end})+\lim_{k\rightarrow{0}}\,({\partial}/{\partial{k}})^{n}AF_{k}(\gammabar_{end})\,\gammabar_{end}^{k}$ which is eventually and readily simplified into \eqref{Eq:HighOrderErgodicCapacityForHighSNRRegime} by using the binomial differentiation rule \cite[Section 0.42]{BibGradshteynRyzhikBook} and \cite[Eq. (5.1.4/s)]{BibZwillingerBook}, which proves \theoremref{Theorem:HighOrderErgodicCapacityForHighSNRRegime}.
\end{proof}

Note that the first term on the right part of \eqref{Eq:HighOrderErgodicCapacityForHighSNRRegime}, i.e., $\scalemath{0.9}{0.9}{\textstyle\log^{n}(\gammabarsub{end})}$ is similar to \eqref{Eq:JensenEqualityBasedApproximatedHOECapacity}, which evidently indicates the slope in log-domain. The second term $\scalemath{0.9}{0.9}{\textstyle\sum_{k=0}^{n}\Binomial{n}{k}\mu_{k}\,\log^{n-k}\left(\gammabar_{end}\right)}$ is related to the diversity gains, and it points out the differences (i.e., constant gaps in log-domain) among different fading environments, and the second term can be obtained in closed form  as shown in the following subsection. Further from numerical point of view, since the $n$th-order derivative can be easily and efficiently accomplished utilizing Grunwald-Letnikov's differentiation \cite{BibGrunwald1867,BibLetnikov1868,BibKilbasSrivastavaTrujillo2006}, we can readily achieve the numerical computation of $\mu_{k}$ as follows 
\begin{equation}\label{Eq:AuxiliaryCoefficientForHighSNRRegimeWithGrunwaldLetnikovDifferentiation}
\mu_{k}=\frac{1}{2\epsilon^{k}}
	\sum_{j=0}^{k}(-1)^{j}\Binomial{k}{j}
		\biggl\{(-1)^{k}AF_{(j\,\epsilon)}(\gammabar_{end})+AF_{(-j\,\epsilon)}(\gammabar_{end})\biggr\},
\end{equation}
where $\epsilon$ is chosen as a small number (e.g. $\epsilon=0.001$). In addition, note that the other accomplishment of \theoremref{Theorem:HighOrderErgodicCapacityForHighSNRRegime} is to provide an asymptotically tight lower-bound for the ACC as shown in the following theorem.

\begin{theorem}\label{Theorem:ErgodicCapacityForHighSNRRegime}
An asymptotically tight lower-bound of the ACC $\widebar{C}(\gammabarsub{end})\equiv\mathbb{E}[\log(1+\gammasub{end})]$
in high-SNR regime of fading environments is given by
\begin{equation}\label{Eq:ErgodicCapacityForHighSNRRegime}
	\widebar{C}(\gammabarsub{end})\gtrapprox\log\left(\gammabarsub{end}\right)+\mu,\text{~~where~~}\mu=\left.\frac{\partial}{\partial{n}}AF_{n}(\gammabar_{end})\right|_{n=0}.
\end{equation}
\end{theorem}

\begin{proof}
Proof is obvious using $\widebar{C}(\gammabarsub{end})\trigeq\widebar{C}(1;\gammabarsub{end})$ with \theoremref{Theorem:HighOrderErgodicCapacityForHighSNRRegime}.
\end{proof}

Consequently, the analysis presented above eliminates the mathematically tedious analysis for the HOCC in high-SNR regimes by means of establishing its sharp characterization of the HOCC using asymptotically tight-bounds. As such, the difference between the actual HOCC and its asymptotic bound is identically equal to zero for $\gammabar_{end}\gg{1}$.

\subsection{Auxiliary tools to obtain the auxiliary coefficient}\label{Sec:SectionIIIA}
Note that exponential-type random distributions such as exponential, gamma, Weibull, generalized gamma, etc. have been widely used in the literature to describe the fading phenomena,  i.e., to statistically characterize the variation in signal strength as a result of multipath propagation. Note that the higher-order moment of exponential-type distributions are mostly in terms of the product of the gamma function $\Gamma\left(a+b\,k\right)$ and the Gaussian function $\exp\left(a+b\,k^{2}\right)$, where $\Gamma\left(\cdot\right)$ denotes the Gamma function \cite[Eq. (6.5.3)]{BibAbramowitzStegunBook}. Referring to \theoremref{Theorem:HighOrderErgodicCapacityForHighSNRRegime}, the $n$th order differentiation of both $\Gamma\left(a+b\,k\right)$ and $\exp\left(a+b\,k^{2}\right)$, i.e.,
\begin{subequations}
\begin{eqnarray}
\label{Eq:HigherOrderDifferentiationOfGammaFunctionWithDerivative}
\Psi_{(n)}\left(a,b,k\right)&=&\left({\partial}/{\partial{k}}\right)^{n}\Gamma\left(a+b\,k\right),\\
\label{Eq:HigherOrderDifferentiationOfGaussianFunctionWithDerivative}
\Phi_{(n)}\left(a,b,k\right)&=&\left({\partial}/{\partial{k}}\right)^{n}\exp\left(a+b\,k^{2}\right),
\end{eqnarray}
\end{subequations}
are strictly required but easily computed by Grunwald-Letnikov's numerical differentiation \cite{BibGrunwald1867,BibLetnikov1868,BibKilbasSrivastavaTrujillo2006}. Furthermore, using built-in functions in the core software coding of such mathematical software packages as \Mathematica, both \eqref{Eq:HigherOrderDifferentiationOfGammaFunctionWithDerivative} and \eqref{Eq:HigherOrderDifferentiationOfGaussianFunctionWithDerivative} are also respectively computed by
\begin{subequations}\label{Eq:AxuiliaryFunctionMathematicaCode}
\begin{eqnarray}
\label{Eq:AxuiliaryFunctionMathematicaCodeA}
\nonumber
\multido{}{4}{\!}&\begin{minipage}{0.94\columnwidth} 
		\definecolor{light-gray}{gray}{0.90}
		\lstset{language=Mathematica,basicstyle=\fontsize{7}{8}\selectfont,breaklines=true,morekeywords={PhiFunction},backgroundcolor=\color{light-gray}}
		\begin{lstlisting}
 (*Implementation for the HOD of the Gamma function :*)
 		\end{lstlisting}
\end{minipage}&\multido{}{10}{~}\\[-7mm]
\multido{}{4}{\!}&\begin{minipage}{0.94\columnwidth} 
		\definecolor{light-gray}{gray}{0.90}
		\lstset{language=Mathematica,basicstyle=\fontsize{7}{8}\selectfont,breaklines=true,morekeywords={PhiFunction},backgroundcolor=\color{light-gray}}
		\begin{lstlisting}
 PhiFunction[n_?Integers, a_, b_, k_] := Derivative[0, 0, n][Function[{u, v, w}, Gamma[u + v w]]][a, b, k];
		\end{lstlisting}
\end{minipage}&\multido{}{10}{~}\\[-3mm]
\label{Eq:AxuiliaryFunctionMathematicaCodeB}
\nonumber
\multido{}{4}{\!}&\begin{minipage}{0.94\columnwidth} 
		\definecolor{light-gray}{gray}{0.90}
		\lstset{language=Mathematica,basicstyle=\fontsize{7}{8}\selectfont,breaklines=true,morekeywords={PhiFunction},backgroundcolor=\color{light-gray}}
		\begin{lstlisting}
 (*Implementation for the HOD of the exponential function :*)
		\end{lstlisting}
\end{minipage}&\multido{}{10}{~}\\[-7mm]
\multido{}{4}{\!}&\begin{minipage}{0.94\columnwidth}
		\definecolor{light-gray}{gray}{0.90}
		\lstset{language=Mathematica,basicstyle=\fontsize{7}{8}\selectfont,breaklines=true,morekeywords={PsiFunction},backgroundcolor=\color{light-gray}}
		\begin{lstlisting}
 PsiFunction[n_?Integers, a_, b_, k_] := Derivative[0, 0, n][Function[{u, v, w}, Exp[u + v w^2]]][a, b, k];
		\end{lstlisting}
\end{minipage}&\multido{}{10}{~}
\end{eqnarray}
\end{subequations}
We note that both \eqref{Eq:AxuiliaryFunctionMathematicaCodeA} and \eqref{Eq:AxuiliaryFunctionMathematicaCodeB} are numerically efficient for small orders $n\in\mathbb{N}$ but erroneous for higher orders. Therefore, their closed forms are evidently required. To the best of our knowledge, $\Psi_{(n)}\left(a,b,k\right)$ and $\Phi_{(n)}\left(a,b,k\right)$ are not given in closed-form expressions in the literature. In the following theorems, we obtain both $\Psi_{(n)}\left(a,b,k\right)$ and $\Phi_{(n)}\left(a,b,k\right)$ in two closed-form expressions by means of using special functions widely used in the literature.

\begin{theorem}\label{Theorem:HighOrderDifferentiationOfGammaFunction}
$\Psi_{(n)}\left(a,b,k\right)\trigeq\left({\partial}/{\partial{k}}\right)^{n}\Gamma\left(a+b\,k\right)$ is given in a closed-form expression as follows 
\begin{equation}\label{Eq:HigherOrderDifferentiationOfGammaFunction}
\Psi_{(n)}\left(a,b,k\right)=
 	\fact{n}\,b^{n}\,
 		\MeijerG[right]{n+2,0}{n+1,n+2}{1}{1,1,\ldots,1}{a+bk,0,0,\ldots,0}
 			+
 			{(-1)}^{n}\,\fact{n}\,b^{n}\,
 				\MeijerG[right]{1,n+1}{n+1,n+2}{1}{1,1,\ldots,1}{a+bk,0,0,\ldots,0},
\end{equation}
where $\MeijerGDefinition{m,n}{p,q}{\cdot}$ denotes the Meijer's G function\emph{\cite[Eq. (9.301)]{BibGradshteynRyzhikBook}}.
\end{theorem}
\vspace{-2mm}

\begin{proof}
Note that using both the upper-incomplete Gamma function $\Gamma\left(\alpha,x\right)\!=\!\int_{x}^{\infty}e^{-u}u^{\alpha-1}du$\cite[Eq. (8.350/1)]{BibGradshteynRyzhikBook} and the lower-incomplete Gamma function $\gamma\left(\alpha,x\right)\!=\!\int_{0}^{x}e^{-u}u^{\alpha-1}du$\cite[Eq. (8.350/2)]{BibGradshteynRyzhikBook} and then utilizing the well-known relation $\Gamma\left(\alpha\right)=\Gamma\left(\alpha,x\right)+\gamma\left(\alpha,x\right)$ \cite[Eq. (8.356/3)]{BibGradshteynRyzhikBook}, we can rewrite $\Psi_{(n)}\left(a,b,k\right)$ as 
\begin{equation}\label{Eq:PsiFunctionUpperLowerIncompleteGammaExpansion}
	\Psi_{(n)}\left(a,b,k\right)=b^{n}\biggl\{\int_{0}^{1}u^{a+b\,k-1}e^{-u}\log^{n}(u)\,du+
		\int_{1}^{\infty}u^{a+b\,k-1}e^{-u}\log^{n}(u)\,du\biggr\}.
\end{equation}
The Meijer's G representation of $\log^{n}(u)$ is given by \cite[Eq. (01.04.26.0046.01)]{BibWolfram2010Book} for $u>1$; and by \cite[Eq. (01.\allowbreak{}04.\allowbreak{}26.\allowbreak{}0045.\allowbreak{}01)]{BibWolfram2010Book} for $0<u\leq{1}$. Note that, referring to \cite[Eq. (9.301)]{BibGradshteynRyzhikBook}, the Mellin-Barnes representations of both Meijer's G functions can be obtained and then substituted into \eqref{Eq:PsiFunctionUpperLowerIncompleteGammaExpansion}. Accordingly,  using \cite[Eqs. (8.350/1), (8.350/2) and (8.356/3)]{BibGradshteynRyzhikBook} after changing the order of the integrals, we rewrite \eqref{Eq:PsiFunctionUpperLowerIncompleteGammaExpansion} in terms of the Mellin-Barnes contour integration, that is
\begin{equation}\label{Eq:PsiFunctionWithMellinBarnesIntegral}
\Psi_{(n)}\left(a,b,k\right)=n!\,b^{n}\biggl\{
	\frac{1}{2\pi\imaginary}\oint_{\mathcal{C}_1}\frac{\Gamma^{n+1}(-s)\Gamma(a+b\,k+s)}{\Gamma^{n+1}(1-s)}ds
		+
		\frac{(-1)^{n}}{2\pi\imaginary}\oint_{\mathcal{C}_2}\frac{\Gamma^{n+1}(s)\Gamma(a+b\,k+s)}{\Gamma^{n+1}(1+s)}ds\biggr\}.
\end{equation}
where the contour integrals $C_1$ and $C_2$ are chosen counter-clockwise in order to ensure the convergence. Finally, using the Mellin-Barnes representation of Meijer's G function\cite[Eq. (9.301)]{BibGradshteynRyzhikBook}, \eqref{Eq:PsiFunctionWithMellinBarnesIntegral} can be represented in terms of Meijer's G function as shown in \eqref{Eq:HigherOrderDifferentiationOfGammaFunction}, which proves \theoremref{Theorem:HighOrderDifferentiationOfGammaFunction}.
\end{proof}

\begin{theorem}\label{Theorem:HighOrderDifferentiationOfExponentialFunction}
$\Phi_{(n)}\left(a,b,k\right)\trigeq\left({\partial}/{\partial{k}}\right)^{n}\exp\left(a+b\,k^{2}\right)$ is given in a closed-form expression as follows
\begin{equation}\label{Eq:HigherOrderDifferentiationOfGaussianFunction}
\Phi_{(n)}\left(a,b,k\right)=
	\pi{e^a}{(-2)^n}b^{n/2}\MeijerG[right]{2,0}{2,3}{b\,k^2}{\frac{1-n}{2},\frac{1-n}{2}}{0,\frac{1}{2},\frac{1-n}{2}}.
\end{equation}
\end{theorem}

\begin{proof}
Note that substituting the Mellin-Barnes representation of Meijer's G function\cite[Eq.(9.301)]{BibGradshteynRyzhikBook} into \cite[Eq. (8.4.3/5)]{BibPrudnikovBookVol3} and then performing algebraic manipulations, we can re-write $\Phi_{(n)}\left(a,b,k\right)$ as
\begin{equation}\label{Eq:PhiFunctionWithSimplifiedMellinBarnesIntegral}
\Phi_{(n)}\left(a,b,k\right)=
	\pi\exp(a){\left(\frac{\partial}{\partial{k}}\right)^{n}\MeijerG{1,0}{1,2}{{b\,k^2}}{{1}/{2}}{0,{1}/{2}}}
	=\pi\exp{a}(-2)^{n}\,b^{{n}/{2}}
	\frac{1}{2\pi\imaginary}\!\oint_{\mathcal{C}}
		\frac{\Gamma(s\frac{n+1}{2})\Gamma(s+\frac{n}{2})(b\,k^2)^{-s}}{\Gamma(\frac{1}{2}-s)\Gamma(\frac{1}{2}+s)\Gamma(\frac{1}{2}+s)}ds,
\end{equation}
which can be represented in terms of Meijer's G function as shown in \eqref{Eq:HigherOrderDifferentiationOfGaussianFunction}, which proves \theoremref{Theorem:HighOrderDifferentiationOfExponentialFunction}.
\end{proof}

With the aid of \cite[Eqs. (9.303) and (9.304)]{BibGradshteynRyzhikBook}, both \eqref{Eq:HigherOrderDifferentiationOfGammaFunction} and \eqref{Theorem:HighOrderDifferentiationOfExponentialFunction} can also be expressed in terms of the sum of generalized hypergeometric functions $\Hypergeom{n}{n+1}{}{\cdot}{\cdot}{\cdot}$ \cite[Eq.~\allowbreak{}(9.14/1)]{BibGradshteynRyzhikBook}, which may be useful for researchers and theoreticians. In addition, referring to \theoremref{Theorem:ErgodicCapacityForHighSNRRegime}, needed for the asymptotic analysis of the ACC in fading environments are both $\Psi_{(1)}\left(a,b,k\right)$ and $\Phi_{(1)}\left(a,b,k\right)$. Accordingly, substituting $n=1$ in \eqref{Eq:HigherOrderDifferentiationOfGammaFunction} and using \cite[Eq.~(8.2.2/3)]{BibPrudnikovBookVol3} and \cite[Eq.~(8.4.51/1)]{BibPrudnikovBookVol3}, we have 
\begin{equation}\label{Eq:FirstOrderDifferentiationOfGammaFunction}
\Psi_{(1)}\left(a,b,k\right)=b\,\Gamma\left(a+b\,k\right)\PolyGamma{0}{a+b\,k},
\end{equation}
where $\PolyGamma{0}{\cdot}$ is the digamma function defined in \cite[Eq. (6.3.1)]{BibAbramowitzStegunBook}. Similarly, substituting $n=1$ in \eqref{Eq:HigherOrderDifferentiationOfGaussianFunction} and using both \cite[Eq.~(8.2.2/15)]{BibPrudnikovBookVol3} and \cite[Eq.~(8.4.3/5)]{BibPrudnikovBookVol3}, we have 
\begin{equation}\label{Eq:FirstOrderDifferentiationOfGaussianFunction}
\Phi_{(1)}\left(a,b,k\right)=2bk\,\exp\left(a+b\,k^2\right).
\end{equation}
In addition, the $k$th order moment of such fading distributions as $\kappa-\mu$ and $\eta-\mu$ distributions consists of generalized hypergeometric function $\Hypergeom{p}{q}{}{a_1,a_2,\ldots,a_p}{b_1,b_2,\ldots,b_p}{x}$ \cite[Eq. (7.2.3/1)]{BibPrudnikovBookVol3} whose parameters $\{a_{\ell}\}_{1}^{p}$, $\{b_{\ell}\}_{1}^{q}$ and $x$ could be a function of the moment order $k$. As a consequence of both \theoremref{Theorem:HighOrderErgodicCapacityForHighSNRRegime} and \theoremref{Theorem:ErgodicCapacityForHighSNRRegime}, the \emph{higher-order differentiation of generalized hypergeometric} (HOD-GH) function is defined in the following theorem.

\begin{theorem}\label{Theorem:HODGHFunction}
The HOD-GH function is defined by
\vspace{-2mm}
\begin{equation}\label{Eq:HigherOrderDerivativeGeneralizedHypergeometricFunction}
\DHypergeom{p}{q}{}{a_1,\ldots,a_p}{b_1,\ldots,b_q}{x}{m_1,\ldots,m_p}{n_1,\ldots,n_q}{k}=
	\Biggl\{\prod_{i=1}^{p}\frac{\partial^{m_i}}{\partial{a^{m_i}_i}}\Biggr\}
		\Biggl\{\prod_{j=1}^{q}\frac{\partial^{n_j}}{\partial{b^{n_j}_j}}\Biggr\}\!
			\Biggl\{\frac{\partial^{k}}{\partial{x^{k}}}\Biggr\}\,
				\Hypergeom{p}{q}{}{a_1,\ldots,a_p}{b_1,\ldots,b_p}{x},	
\end{equation}
where for all $i\in\{0,1,2,\ldots,p\}$, $j\in\{1,2,\ldots,q\}$, the values $m_{i},n_{j},k\in\mathbb{N}$ denote the differentiation orders.
\end{theorem}

\begin{proof}
The proof is obvious using some special cases \cite[Eqs. (07.31.20.0009.01) and (07.31.20.0010.01)]{BibWolfram2010Book}.
\end{proof}

We find that the HOD-GH function (i.e., \eqref{Eq:HigherOrderDerivativeGeneralizedHypergeometricFunction}) is fortunately implemented, but not well-documented, in the core software coding of \Mathematica \cite{BibWolfram2010Book}, and further its traditional form is specifically introduced in \Mathematica as
${}_{p}{F}_{q}{}^{
    \scalemath{0.6}{0.6}{(\{m_1,\ldots,m_p\},\{n_1,\ldots,n_q\},k)}
    }\!\left[{a_1,\ldots,a_p}\,{;}\,{b_1,\ldots,b_q}\,{;}\,{x}\right]$.
Accordingly, both the implementation and the example usage of the HOD-GH function is nicely presented as 
\begin{equation}\label{Eq:MathematicaHODForGeneralizedHypergeometricFunction}
	\begin{minipage}{0.94\textwidth}
		\definecolor{light-gray}{gray}{0.90}
		\lstset{language=Mathematica,basicstyle=\fontsize{7}{8}\selectfont,breaklines=true,morekeywords={HODHypergeometricPFQ},backgroundcolor=\color{light-gray}}
\begin{lstlisting}
 (*Implementation of the HOD-GH function :*)
 HODHypergeometricPFQ[m_, n_, k_][a_, b_, x_] := Derivative[da, db, dx][Function[{u, v, w}, HypergeometricPFQ[u, v, w]]][a, b, x];

 (*Example usage :*)
  In[1]:= HODHypergeometricPFQ[{1, 0}, {1}, 0][{1, 2}, {3}, 0.5]] // N
 Out[1]:= -0.3795244284791705
\end{lstlisting}
	\end{minipage}
\end{equation}
which is numerically efficient and symbolically useful for a wide range of scientific computing and analysis. Consequently, with the aid of all computable tools presented above, we can efficiently compute  \eqref{Eq:HighOrderErgodicCapacityForHighSNRRegime} and so \eqref{Eq:ErgodicCapacityForHighSNRRegime}. Therefore, let us consider some special cases in order to check analytical simplicity and accuracy of both \theoremref{Theorem:HighOrderErgodicCapacityForHighSNRRegime} and \theoremref{Theorem:ErgodicCapacityForHighSNRRegime}.

\subsubsection{HOCC in high-SNR regime of generalized Nakagami-\emph{m} fading environments}
\label{Example:HigherOrderCapacityOverGeneralizedGammaFadingChannel}
In generalized Nakagami-\emph{m} fading environments \cite{BibYilmazAlouiniGLOBECOM2009}, $\gamma_{end}$ follows the PDF given by
\begin{equation}\label{Eq:EqGeneralizedNakagamiMPDF}
p_{\gammasub{end}}\left(\gamma;\gammabarsub{end}\right)=\frac{\xi}{\Gamma\left(m\right)}
	\left(\frac{\beta}{\gammabarsub{end}}\right)^{\xi{m}}
			\gamma^{\xi{m}-1}
				\exp\biggl({\displaystyle-\left(\frac{\beta}{\gammabarsub{end}}\right)^{\xi}\gamma^{\xi}}\biggr),
\end{equation}
where $\gammabar_{end}$ denotes the average power (i.e., $\gammabar_{end}=\mathbb{E}[\gamma_{end}]$) as mentioned before. Further, $m$ and $\xi$ denote the fading figure $(0.5\leq{m})$ and the shaping parameter $(\xi>0)$, respectively; accordingly, $\beta=\Gamma\left(m+1/\xi\right)/\Gamma\left(m\right)$. The special or limiting cases of the generalized Nakagami-\emph{m} distribution are well-known in the literature as the Rayleigh $(m=1,\,\xi=1)$, Half-Normal $(m=1/2,\,\xi=1)$, Nakagami-\emph{m} $(\xi=1)$, Weibull $(m=1)$, lognormal $(m\rightarrow\infty,\,\xi\rightarrow{0})$, and AWGN $(m\rightarrow\infty,\,\xi=1)$. In order to use \theoremref{Theorem:HighOrderErgodicCapacityForHighSNRRegime}, the higher-order AOF $AF_n(\gamma_{end})$ is easily attained by using the $n$th-moment of the generalized gamma distribution obtained in \cite[Eq.~(23) with the special case $N=1$]{BibYilmazAlouiniGLOBECOM2009}, that is 
\begin{equation}\label{Eq:AmountOfFadingForGeneralizedGammaFading}
	AF_{n}(\gamma_{end})=\frac{\Gamma\left(m+{n}/{\xi}\right)}{\Gamma(m)}{\beta}^{-n}-1.
\end{equation}
Next, substituting \eqref{Eq:AmountOfFadingForGeneralizedGammaFading} into \theoremref{Theorem:HighOrderErgodicCapacityForHighSNRRegime} 
and using the binomial differentiation rule \cite[Sec.~0.42]{BibGradshteynRyzhikBook} with 
\eqref{Eq:HigherOrderDifferentiationOfGammaFunctionWithDerivative}, we have 
\begin{equation}\label{Eq:AuxiliaryCoefficientForHighSNRRegimeForGeneralizedGammaFading}
\mu_k=\sum_{j=0}^{k}
	{(-1)}^{j}\Binomial{k}{j}\,
		\Psi_{(k-j)}\left({\textstyle{m},{1}/{\xi},0}\right)\,\log^{j}\left(\beta\right)-\KroneckerDelta{k}{0},
\end{equation}
where $\KroneckerDelta{i}{j}$ is the Kronecker's delta (i.e, $\KroneckerDelta{i}{j}=1$ if $j=j$, and $0$ otherwise) \cite{BibAbramowitzStegunBook}. Accordingly, substituting \eqref{Eq:AuxiliaryCoefficientForHighSNRRegimeForGeneralizedGammaFading} into \eqref{Eq:HighOrderErgodicCapacityForHighSNRRegime}, we asymptotically estimate the HOCC $\widebar{C}(n;\gammabarsub{end})$ in high-SNR regime as
\begin{equation}\label{Eq:HighOrderErgodicCapacityForHighSNRRegimeForGeneralizedGammaFading}
\widebar{C}(n;\gammabarsub{end})\gtrapprox
	\sum_{k=0}^{n}
	\Binomial{n}{k}
	\log^{n-k}\left(\gammabarsub{end}\right)
	~\sum_{j=0}^{k}{(-1)}^{j}\Binomial{k}{j}\,\Psi_{(k-j)}\left({\textstyle{m},{1}/{\xi},0}\right)\,\log^{j}\left(\beta\right),
\end{equation}
where substituting $n=1$ yields 
\begin{equation}\label{Eq:GeneralizedGammaErgodicCapacityForHighSNRRegime}
\widebar{C}(\gammabarsub{end})\gtrapprox\log\left(\gammabarsub{end}\right)-\log\left(\beta\right)+{\PolyGamma{0}{m}}/{\xi},
\end{equation}
which is the ACC in high-SNR regime of generalized gamma fading environment. Note that \eqref{Eq:GeneralizedGammaErgodicCapacityForHighSNRRegime} is a simpler expression as compared to \cite[Eq. (18)]{BibYilmazAlouiniGLOBECOM2011}. Further, setting the fading figure to an integer value (i.e., $m\in\mathbb{Z}^{+}$) and then using \cite[Eq. (6.3.2)]{BibAbramowitzStegunBook}, the asymptotic ACC simplifies to a much simpler expression, i.e., $\widebar{C}(\gammabarsub{end})\gtrapprox\log\left(\gammabarsub{end}\right)-\log\left(\beta\right)+\frac{1}{\xi}\sum_{k=1}^{m-1}\frac{1}{k}-\frac{\EulerGamma}{\xi}$, where $\EulerGamma=0.5772156649...$ is Euler-Mascheroni constant \cite{BibAbramowitzStegunBook}.

For consistency, let us consider some special cases of generalized Nakagami-\emph{m} fading environments. Accordingly, substituting $m=1$ into \eqref{Eq:GeneralizedGammaErgodicCapacityForHighSNRRegime}, we obtain the asymptotic ACC in high-SNR regime of Weibull fading environments, that is
\begin{equation}\label{Eq:WeibullErgodicCapacityForHighSNRRegime}
	\widebar{C}(\gammabarsub{end})\gtrapprox\log\left(\gammabarsub{end}\right)-\log\left(\Gamma(1+{1}/{\xi})\right)-{\EulerGamma}/{\xi}.
\end{equation}
Similarly, substituting $\xi=1$ into \eqref{Eq:GeneralizedGammaErgodicCapacityForHighSNRRegime}, we obtain the asymptotic ACC in Nakagami-\emph{m} fading environments, that is 
\begin{equation}\label{Eq:NakagamiErgodicCapacityForHighSNRRegime}
	\widebar{C}(\gammabarsub{end})\gtrapprox\log\left(\gammabarsub{end}\right)-\log\left(m\right)+\PolyGamma{0}{m}
		\text{ for $m\in\mathbb{R}^{+}$, and~~}
	\widebar{C}(\gammabarsub{end})\gtrapprox\log\left(\gammabarsub{end}\right)-\log\left(m\right)+{\textstyle\sum_{k=1}^{m-1}\frac{1}{k}}-\EulerGamma
		\text{ for $m\in\mathbb{Z}^{+}$.}
\end{equation}
Accordingly, the asymptotic ACC in Rayleigh fading environments ($m=1,\,\xi=1$) is written as $\widebar{C}(\gammabarsub{end})\gtrapprox\log\left(\gammabarsub{end}\right)-\EulerGamma$. The other important special case  is the AWGN $(m\rightarrow\infty,\,\xi=1)$. Then, substituting $m\rightarrow\infty$ and $\xi=1$ in \eqref{Eq:GeneralizedGammaErgodicCapacityForHighSNRRegime} and using the approximation $\lim_{m\rightarrow\infty}\PolyGamma{0}{m}\approx\sum_{k=1}^{m-1}\frac{1}{k}-\EulerGamma\approx\log(m)$ that is obtained with the aid of \cite[Eq.~(4.1.32)]{BibAbramowitzStegunBook}, the ACC of the AWGN channels is obtained for high-SNR regime as $\widebar{C}(\gammabarsub{end})\gtrapprox\log\left(\gammabarsub{end}\right)$ as expected.


\subsubsection{HOCC in high-SNR regime of lognormal fading environments}
\label{Example:HigherOrderCapacityOverLognormalFadingChannel}
In lognormal fading environments, $\gammasub{end}$ follows the PDF given by \cite[Eq. (2.53)]{BibAlouiniBook}
\begin{equation}\label{Eq:EqLognormalPDF}
\!\!p_{\gammasub{end}}\left(\gamma;\gammabarsub{end}\right)=\frac{\kappa}{\sqrt{2\pi}\,\sigma\,\gamma}\exp\biggl(-\frac{{(\kappa\,\log(\gamma)-\mu)}^{2}}{2\,\sigma^2}\biggr),
\end{equation}
where $\kappa=10/\log(10)$, and where $\mu$\,(dB) and $\sigma$\,(dB) denote the mean and the standard deviation of $\gammasub{end}$, respectively. The $n$th order moment of lognormal fading distribution is given by  $\Expected{\gamma_{end}^{n}}=\exp\bigl(\frac{\mu}{\kappa}n\bigr)\exp\bigl(\frac{\sigma^2}{2\kappa^2}n^2\bigr)$ \cite[Eq~(2.55)]{BibAlouiniBook}. With this result, we readily obtain the higher-order AOF $AF_{n}(\gamma_{end})$ as
\begin{equation}\label{Eq:HigherOrderAmountOfFadingForLognormalFading}
AF(n)=\exp\left(-\frac{\sigma^2}{2\kappa^2}n\right)\exp\left(\frac{\sigma^2}{2\kappa^2}\,n^2\right)-1,
\end{equation}
After substituting \eqref{Eq:HigherOrderAmountOfFadingForLognormalFading} into \theoremref{Theorem:HighOrderErgodicCapacityForHighSNRRegime} and then performing some algebraic manipulations utilizing \eqref{Eq:HigherOrderDifferentiationOfGaussianFunctionWithDerivative}, we derive the auxiliary coefficient $\mu_{k}$ as
\begin{equation}\label{Eq:AuxiliaryCoefficientForHighSNRRegimeForLognormalFading}
\!\!\!\mu_k=\sum_{j=0}^{k}{(-1)}^{j}\Binomial{k}{j}\,
\Phi_{(k-j)}\!\left({\textstyle{j\log\left(\frac{\sigma^2}{2\kappa^2}\right)},\frac{\sigma^2}{2\kappa^2},0}\right)-\KroneckerDelta{k}{0}.\!
\end{equation}
Accordingly, the HOCC $\widebar{C}(n;\gammabarsub{end})$ for high-SNR regime can be asymptotically estimated as
\begin{equation}\label{Eq:HighOrderErgodicCapacityForHighSNRRegimeForLognormalFading}
\widebar{C}(n;\gammabarsub{end})\gtrapprox
	\sum_{k=0}^{n}
	\Binomial{n}{k}
	\log^{n-k}\left(\gammabarsub{end}\right)
	\sum_{j=0}^{k}{(-1)}^{j}\Binomial{k}{j}\,\Phi_{(k-j)}\!\left({\textstyle{j\log\left(\frac{\sigma^2}{2\kappa^2}\right)},\frac{\sigma^2}{2\kappa^2},0}\right).
\end{equation}
Eventually, either setting $n=1$ in \eqref{Eq:HighOrderErgodicCapacityForHighSNRRegimeForLognormalFading} or utilizing \theoremref{Theorem:ErgodicCapacityForHighSNRRegime} with \eqref{Eq:HigherOrderAmountOfFadingForLognormalFading}, the ACC is obtained as
\begin{equation}\label{Eq:LognormalErgodicCapacityForHighSNRRegime}
	\widebar{C}(\gammabarsub{end})\gtrapprox\log\left(\gammabarsub{end}\right)-\frac{\sigma^2}{2\kappa^2}.
\end{equation}
Note that, as seen in \eqref{Eq:LognormalErgodicCapacityForHighSNRRegime}, the ACC for high-SNR regime in lognormal fading environments does not seem to depend on the mean parameter $\mu$ but it is hidden in $\gammabar_{end}$. In more details,
$\gammabar_{end}$ strictly depends on both $\mu$ and $\sigma$; so given by $\gammabar_{end}=\exp\bigl(\frac{\mu}{\kappa}\bigr)\exp\bigl(\frac{\sigma^2}{2\kappa^2}\bigr)$\cite[Eq~(2.55)]{BibAlouiniBook}.
Then, substituting this $\gammabar_{end}$ into \eqref{Eq:LognormalErgodicCapacityForHighSNRRegime}, the ACC in high-SNR regime of lognormal fading environments becomes $\widebar{C}(\gammabar_{end})\gtrapprox\frac{\mu}{\kappa}$ which itself not depend on the deviation parameter $\sigma$.

\subsubsection{HOCC in high-SNR regime of extended generalized-K (EGK) fading environments}
\label{Example:HigherOrderCapacityOverEGKFadingChannel} The extended generalized-K (EGK) distribution, which is proposed in \cite{BibYilmazAlouiniISWCS2010,BibYilmazAlouiniEGK2010}, is such a distribution that several fading distributions are either its special or limiting cases such as Rayleigh, lognormal, Weibull, Nakagami-\emph{m}, generalized Nakagami-\emph{m}, generalized-K and the others listed in \cite[Table~I]{BibYilmazAlouiniISWCS2010},\cite[Table~I]{BibYilmazAlouiniEGK2010}. Regarding this disclosed versatility, the EGK distribution offers a unified theory to statistically characterize the envelope statistics of known wireless / optical communication channels. In EGK fading environments, $\gamma_{end}$ follows the PDF given by \cite[Eq.(3)]{BibYilmazAlouiniISWCS2010}, \cite[Eq.(5)]{BibYilmazAlouiniEGK2010}
\begin{equation}\label{Eq:EqInstantaneousSNRPDFForEGKFading}
p_{\gammasub{end}}\left(\gamma;\gammabarsub{end}\right)=
	\frac{\xi}{\Gamma(\msub{s})\Gamma(m)}
		{\left(\frac{\betasub{s}\beta}{\gammabarsub{end}}\right)}^{{m}\xi}
			{\gamma}^{{m}\xi-1}
					\ExtGamma{\msub{s}-{m}\frac{\xi}{\xisub{s}}}{0}{\left(\frac{\betasub{s}\beta}{\gammabarsub{end}}\right)^{m\xi}{\gamma}^{\xi}}{\frac{\xi}{\xisub{s}}},
\end{equation}
where $m~(0.5\leq{m}<\infty)$ and $\xi~(0\leq\xi<\infty)$ denote the fading figure (diversity severity / order) and the fading shaping factor, respectively. Further, $\msub{s}~(0.5\leq\msub{s}<\infty)$ and $\xisub{s}~(0\leq\xi<\infty)$ denote the shadowing severity and the shadowing shaping factor (homogeneity), respectively. Accordingly, $\beta$ and $\betasub{s}$ are defined as $\beta=\Gamma\left(m+1/\xi\right)/\Gamma\left(m\right)$ and $\betasub{s}=\Gamma\left(\msub{s}+1/\xisub{s}\right)/\Gamma\left(\msub{s}\right)$, respectively. 
In \eqref{Eq:EqInstantaneousSNRPDFForEGKFading}, $\ExtGamma{\cdot}{\cdot}{\cdot}{\cdot}$ is the extended incomplete Gamma function defined as $\ExtGamma{\alpha}{x}{b}{\beta}=\int_{x}^{\infty}r^{\alpha-1}\exp\left(-{r}-b{r}^{-\beta}\right)dr$, where $\alpha,\beta,b\in\mathbb{C}$ and $x\in\mathbb{R}^{+}$\cite[Eq. (6.2)]{BibChaudhryZubairBook}. 

\begin{lemma}[Multinomial Differentiation Rule]\label{Lemma:MultinomialDifferentiationRule}
Let $\{F_{n}(x)\}$ be a function set of size $N$, each of which is continuous and the $k$-order differentiable over $x\in\mathbb{C}$. The $k$th derivative of the product of $\prod_{n=1}^{N}F_{n}(x)$ is given by
\small
\begin{equation}
\biggl(\frac{\partial}{\partial{x}}\biggr)^{k}\!\prod_{n=1}^{N}F_{n}(x)=\!\!
	\sum_{j_1+j_2+\ldots+j_N=k}^{k}
		\Multinomial{k}{j_1,j_2,\ldots,j_N}
			\prod_{n=1}^{N}
				\biggl[\biggl(\frac{\partial}{\partial{x}}\biggr)^{j_n}\!F_{j_n}(x)\biggl].
\end{equation}
\end{lemma}

\begin{proof}
The proof is obvious which follows by induction from Leibnitz's rule.
\end{proof}
 
Note that the $n$th order moment $\Expected{\gamma_{end}^{n}}$ is given in 
\cite[Eq.~(9)]{BibYilmazAlouiniEGK2010}. Accordingly, we obtain the higher-order AOF $AF_{n}(\gamma_{end})$ as follows
\begin{equation}\label{Eq:AmountOfFadingForEGKFading}
AF_{n}(\gamma_{end})=\frac{\Gamma\left(m+{n}/{\xi}\right)\Gamma\left(m_s+{n}/{\xi_s}\right)}{\Gamma(m)\Gamma(m_s)}\left(\beta\,\beta_s\right)^{-n}-1.
\end{equation}
In order to obtain the auxiliary coefficient $\mu_{k}$, its higher-order differentiation can be obtained by using the multinomial differentiation rule given in the lemma above. Substituting \eqref{Eq:AmountOfFadingForEGKFading} into \theoremref{Theorem:HighOrderErgodicCapacityForHighSNRRegime} and then exploiting \lemmaref{Lemma:MultinomialDifferentiationRule}, we derive the auxiliary coefficient $\mu_{k}$ as
\begin{equation}\label{Eq:AuxiliaryCoefficientForHighSNRRegimeForEGKFading}
\mu_{k}=\!\!\!\!\sum_{i+j+l=k}^{k}
	{(-1)}^{l}\,\Multinomial{k}{i,j,l}\,\Psi_{(i)}\!\left({\textstyle{m},{1}/{\xi},0}\right)
		\Psi_{(j)}\!\left({\textstyle{m_s},{1}/{\xi_s},0}\right)
			\,\log^{l}\!\left(\beta\,\beta_s\right)-\KroneckerDelta{k}{0},
\end{equation}
where the parameters $\beta$ and $\beta_s$ are mentioned before. Accordingly, substituting \eqref{Eq:AuxiliaryCoefficientForHighSNRRegimeForEGKFading} into \eqref{Eq:HighOrderErgodicCapacityForHighSNRRegime}, the HOCC in high-SNR regime can be accurately estimated as
\begin{equation}\label{Eq:HighOrderErgodicCapacityForHighSNRRegimeForEGKFading} 
\widebar{C}(n;\gammabarsub{end})\gtrapprox
	\sum_{k=0}^{n}
	\Binomial{n}{k}
	\log^{n-k}\left(\gammabarsub{end}\right)
	\sum_{i+j+l=k}^{k}
	{(-1)}^{l}\,\Multinomial{k}{i,j,l}\,
		\Psi_{(i)}\!\left({\textstyle{m},{1}/{\xi},0}\right)\,
		\Psi_{(j)}\!\left({\textstyle{m_s},{1}/{\xi_s},0}\right)\,
			\log^{l}\left(\beta\,\beta_s\right).
\end{equation}
Following the same steps in both \secref{Example:HigherOrderCapacityOverGeneralizedGammaFadingChannel} and \secref{Example:HigherOrderCapacityOverLognormalFadingChannel}, we can accurately estimate the ACC $\widebar{C}(\gammabarsub{end})$ in high-SNR regime of the EGK fading environments, that is
\begin{equation}\label{Eq:EGKErgodicCapacityForHighSNRRegime}
	\widebar{C}(\gammabarsub{end})\gtrapprox\log\left(\gammabarsub{end}\right)-\log\left(\beta\beta_s\right)+{\PolyGamma{0}{m}}/{\xi}+{\PolyGamma{0}{m_s}}/{\xi_s}.
\end{equation}
In addition, for the special or limiting cases listed in both \cite[Table~I]{BibYilmazAlouiniISWCS2010} and \cite[Table~I]{BibYilmazAlouiniEGK2010}, the ACC in High-SNR regime can be easily obtained by means of \eqref{Eq:EGKErgodicCapacityForHighSNRRegime}. For instance, for  $\xi=1$ and $\xi_{s}=1$, \eqref{Eq:EGKErgodicCapacityForHighSNRRegime} simply simplifies to
\begin{equation}\label{Eq:GeneralizedKErgodicCapacityForHighSNRRegime}
\widebar{C}(\gammabarsub{end})\gtrapprox\log\left(\gammabarsub{end}\right)-\log\left(m\,m_s\right)+{\PolyGamma{0}{m}}+{\PolyGamma{0}{m_s}},
\end{equation}
which is the ergodic capacity of the generalized-K fading channels for high-SNR regime.

\subsubsection{HOCC in high-SNR regime of \texorpdfstring{$\kappa-\mu$} fading environments}
\label{Example:HigherOrderCapacityOverKappaMuFadingChannel}
In the case of line-of-sight environment, the $\kappa-\mu$ distribution can be used to statistically characterize the small-scale variation of the fading signal. In $\kappa-\mu$ fading environments, the instantaneous SNR $\gamma_{end}$ follows the PDF \cite{BibYacoub2007}
\begin{equation}\label{Eq:EqInstantaneousSNRPDFForKappaMuFading}
p_{\gamma_{end}}\left(\gamma;\gammabarsub{end}\right)=\frac{\mu{(\kappa+1)}^{\frac{\mu+1}{2}}\gamma^{\frac{\mu-1}{2}}}
	{\kappa^{\frac{\mu-1}{2}}\exp\left(\mu\kappa\right)\gammabar_{end}^{\frac{\mu+1}{2}}}
		\exp\biggl(-\frac{\mu(\kappa+1)}{\gammabarsub{end}}\gamma\biggr)\,
			\BesselI[\mu-1]{2\mu\sqrt{\frac{\kappa(\kappa+1)\gamma}{\gammabar_{end}}}},
\end{equation}
where $\kappa$ denotes the ratio of the total power of the dominant component to that of scattered components, and $\mu$ denotes the line-of-sight parameter defined as $\mu=(2\kappa+1)/({(\kappa+1)}^{2}\,AF_{1}(\gamma_{end}))$. Moreover, $\BesselI[n]{\cdot}$ denotes the modified Bessel function of the first kind of the order $n$ \cite[Eq. (9.6.10)]{BibAbramowitzStegunBook}. Evidently, following the same steps in the previous examples and then using the $n$th moment $\Expected{\gamma_{end}^{n}}$ given in \cite[Eq.~(5)]{BibYacoub2007}, we obtain the higher-order AOF $AF_{n}(\gamma_{end})$ as
\begin{equation}\label{Eq:AmountOfFadingForKappaMuFading}
AF(n)=\frac{\Gamma(\mu+n)}{(\kappa+1)^{n}\mu^{n}}
	\Hypergeom{1}{1}{}{-n}{\mu}{-\kappa\mu}-1,
\end{equation}
where $\Hypergeom{1}{1}{}{\cdot}{\cdot}{\cdot}$ denotes the Kummer's confluent hypergeometric function \cite[Eq. (13.1.2)]{BibAbramowitzStegunBook}. After substituting \eqref{Eq:AmountOfFadingForKappaMuFading} into \eqref{Eq:ErgodicCapacityForHighSNRRegime} and utilizing both \lemmaref{Lemma:MultinomialDifferentiationRule} and \theoremref{Theorem:HODGHFunction}, the auxiliary coefficient $\mu_k$ is readily obtained as \begin{equation}\label{Eq:AuxiliaryCoefficientForHighSNRRegimeForKappaMuFading}
\mu_{k}=\!\!\!\!\sum_{i+j+l=k}^{k}
	{(-1)}^{j+l}\Multinomial{k}{i,j,l}\,
		\Psi_{(i)}\!\left({\mu,1,0}\right)\, 
			\log^{j}\left(\kappa\mu+\mu\right)\,
				\DHypergeomSMALL{1}{1}{}{0}{\mu}{-\kappa\mu}{l}{0}{0}-\KroneckerDelta{k}{0}.
\end{equation}
Accordingly, substituting \eqref{Eq:AuxiliaryCoefficientForHighSNRRegimeForKappaMuFading} into \eqref{Eq:HighOrderErgodicCapacityForHighSNRRegime}, the HOCC $\widebar{C}(n;\gammabarsub{end})$ for high-SNR regime can be accurately estimated as
\begin{equation}\label{Eq:HighOrderErgodicCapacityForHighSNRRegimeForKappaMuFading}
\widebar{C}(n;\gammabarsub{end})\fallingdotseq
	\sum_{k=0}^{n}
	\Binomial{n}{k}
	\log^{n-k}\left(\gammabarsub{end}\right)
		\sum_{i+j+l=k}^{k}
			{(-1)}^{j+l}
			\Multinomial{k}{i,j,l}\,
				\Psi_{(i)}\left(\mu,1,0\right)\,
					\log^{j}\left(\kappa\mu+\mu\right)\,
						\DHypergeom{1}{1}{}{0}{\mu}{-\kappa\mu}{l}{0}{0}.
\end{equation}
With the aid of this result, the ACC in high-SNR regime of the $\kappa-\mu$ fading environment is obtained as
\begin{equation}\label{Eq:KappaMuErgodicCapacityForHighSNRRegime}
\widebar{C}(\gammabarsub{end})\gtrapprox\log\left(\gammabarsub{end}\right)-\log\left((\kappa+1)\mu\right)+
		\PolyGamma{0}{\mu}-\Gamma(\mu)\,\DHypergeomSMALL{1}{1}{}{0}{\mu}{-\kappa\mu}{1}{0}{0},
\end{equation}
where $\DHypergeomSMALL{1}{1}{}{0}{\mu}{-\kappa\mu}{1}{0}{0}$ can be given in terms of series expansion of $\Hypergeom{1}{1}{}{a}{b}{x}$, specifically using \cite[Eq. (13.1.2)]{BibAbramowitzStegunBook}. With that representation, it is clear that $\lim\nolimits_{\kappa\rightarrow{0}}\DHypergeomSMALL{1}{1}{}{0}{\mu}{-\kappa\mu}{1}{0}{0}=0$. In the following, substituting $\kappa\rightarrow{0}$ and $\mu=m$ into \eqref{Eq:KappaMuErgodicCapacityForHighSNRRegime} results in \eqref{Eq:NakagamiErgodicCapacityForHighSNRRegime} which is the ergodic capacity for high-SNR regime in Nakagami-\emph{m} fading environments.
\CHECKQEDsymbol

Consequently, the closed-form expressions, presented in this section, demonstrate the usefulness of \theoremref{Theorem:HighOrderErgodicCapacityForHighSNRRegime} in the asymptotic analysis of the HOCC in high-SNR regime of generalized fading environments, and also show that the auxiliary coefficient $\mu_k$ in \eqref{Eq:HighOrderErgodicCapacityForHighSNRRegime} and the parameter $\mu$ in \eqref{Eq:ErgodicCapacityForHighSNRRegime} can be easily derived using simple functions.

\section{HOCC in low-SNR regime of generalized fading environments}\label{Sec:SectionIV}
In this section, we deal with the sharp characterization of the HOCC in low-SNR regime of generalized fading environments. In low-SNR regime (i.e., $\gammabarsub{end}\ll{1}$), we have $1-P_{\gammasub{end}}(\gamma_{th})\approx{0}$ for $\gamma_{th}\gg\gammabar_{end}$, and therefore using \eqref{Eq:HOCC}, we can 
\begin{wrapfigure}{r}{0.51\textwidth}
  \begin{center}
  	\vspace{-8mm}
	\psfrag{XLabel}[c]{\footnotesize $\gammabar_{end}$}
	\psfrag{YLabel}[c]{\footnotesize Objective Function $\sum_{n=1}^{N}w_{n}\Delta^{2}(n+1;\gammabarsub{end})$}
	\includegraphics[width=0.55\textwidth,keepaspectratio=true]{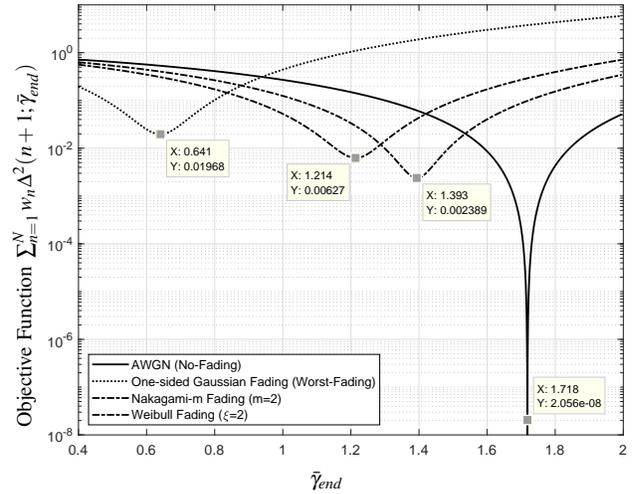}
	\vspace{-6mm}
	\caption{The boundary SNR for different fading environments.}
	\label{Fig:BoundaryTresholdSNRForLowSNRRegime}
	\vspace{-5mm}
  \end{center}
  \vspace{-10mm}
\end{wrapfigure} 
approximate the HOCC $\widebar{C}(n;\gammabarsub{end})$ by means of 
\begin{equation}\label{Eq:ApproximatedHOECapacityForLowSNRRegime}
	\widebar{C}(n;\gammabarsub{end})\gtrapprox\int_{0}^{\gamma_{th}}\log^{n}\left(1+\gamma\right)\,p_{\gamma_{end}}\left(\gamma;\gammabarsub{end}\right)d\gamma,
\end{equation}
where $\gamma_{th}$ is such a average SNR threshold that the low-SNR regime occurs for $\gammabar_{end}<\gamma_{th}$. As the order $n$ increases, $\widebar{C}(n;\gammabarsub{end})$ remains almost the same, as contrary in high-SNR region. This behaviour is described by
\begin{equation}\label{Eq:ChangeInHigherOrderErgodicCapacity}
	\Delta(n;\gammabarsub{end})=\frac{\widebar{C}(n;\gammabarsub{end})}{\widebar{C}(\gammabarsub{end})}-1,
\end{equation}
which is a useful tool to distinguish the high-SNR and low-SNR regimes of the fading environments. Evidently, when $\Delta(n;\gammabarsub{end})\!<\!0$, the fading environment is certainly in low-SNR regime, otherwise, in high-SNR regime. We can ascertain a specific average SNR at which $\Delta(n;\gammabarsub{end})$ approaches zero, particularly unfolding $\gamma_{th}$ between the high-SNR regime and the low-SNR regime. Hence, we indeed need to optimize all objective functions $\{\Delta^{2}(n+1;\gammabarsub{end})\}_{n=1}^{N}$ with respect to $\gammabarsub{end}$, where it is especially chosen as $N=4$ since the first four statistics including mean, variance, skewness, and kurtosis are quite enough for a statistical characterization. Accordingly, referring \emph{multi-objective optimization} (MOO) \cite{BibBurkeKendallBook,BibMarlerArora2004,BibCaramiaDellOlma2008}, a feasible (i.e., Pareto optimal) solution of $\gamma_{th}$ is determined as
\begin{equation}\label{Eq:BoundarySNRValueBetweenHighLowSNRRegimes}
	\gamma_{th}=\underset{\gammabarsub{end}}{\argmin}\,\sum_{n=1}^{N}w_{n}\,\Delta^{2}(n+1;\gammabarsub{end})
			   =\underset{\gammabarsub{end}}{\argmin}\,\sum_{n=1}^{N}w_{n}\,\Bigl(\,\frac{\widebar{C}(n;\gammabarsub{end})}{\widebar{C}(\gammabarsub{end})}-1\Bigr)^{2},
\end{equation}
where the weights $\{w_{n}\}_{n=1}^{N}$ are positive real numbers chosen as $\sum_{n=1}^{N}w_{n}=1$. Referring to Jensen's inequality \cite[Section 12.411]{BibGradshteynRyzhikBook}, the optimization problem apparently suggests that $\mathbb{E}[\log^{n}(1+\gammasub{end})]\approx\mathbb{E}[\log(1+\gammasub{end})]$ for all $n\in\mathbb{N}$ when $\gammabarsub{end}=\gamma_{th}$ for any type of fading environment. The objective function has been depicted in \figref{Fig:BoundaryTresholdSNRForLowSNRRegime} for different fading environments. and the data-tips in \figref{Fig:BoundaryTresholdSNRForLowSNRRegime} provide a convenient way to view information about the Pareto optimal values of $\gamma_{th}$ for the considered fading environments. As it is apparently observed, $\gamma_{th}$ is distinctive for different fading environments. Frankly speaking, $\gamma_{th}$ decreases when the channel quality gets worse (i.e., the diversity order of the channel decreases) or it increases when the channel quality gets better. When there is no fading (i.e., when the corresponding wireless environment is Gaussian), $\gamma_{th}$ reaches its supremum-limit value. Conversely, when the quality of the fading environment gets the worst, $\gamma_{th}$ reaches to its infinum-limit value. 
\begin{theorem}[Supremum-Limit Value of $\gamma_{th}$]\label{Theorem:SupremumLimitSNRValueThatLowSNRRegimeStarts}
Any fadingless environment leads in low-SNR regime when  
\begin{eqnarray}\label{Eq:SupremumLimitSNRValueThatLowSNRRegimeStarts}
	\gammabar_{end}\lessapprox\gamma_{th}^{+}=1.71828182846~~(2.35094397275~\text{\rm{dB}}).
\end{eqnarray}
\end{theorem}

\begin{proof}
When the quality of the fading environment significantly gets better, we have $p_{\gammasub{end}}(\gamma;\gammabarsub{end})\approx\DiracDelta{\gamma-\gammabarsub{end}}$, where $\DiracDelta{\cdot}$ denotes the Dirac's delta function \cite[Eq.\!~(1.8.1)]{BibZwillingerBook}. Referring to \eqref{Eq:HOCC} and then using \cite[Eq.\!~(1.8.1/1)]{BibZwillingerBook}, we have $\widebar{C}(n;\gammabarsub{end})\approx\log^{n}(1+\gammabarsub{end})$, which convert \eqref{Eq:BoundarySNRValueBetweenHighLowSNRRegimes} to a joint solution of a set of equations $\log^{n}(1+\gamma^{+}_{th})=\log(1+\gamma^{+}_{th})$ for all $n\in\mathbb{Z}^{+}$ and ${n}\neq{1}$. For the first four statistics, this equation set has obviously one unique solution which is obtained as $\gamma_{th}^{+}=1.71828182846$ as shown in \eqref{Eq:SupremumLimitSNRValueThatLowSNRRegimeStarts} and depicted in \figref{Fig:BoundaryTresholdSNRForLowSNRRegime}, which proves \theoremref{Theorem:SupremumLimitSNRValueThatLowSNRRegimeStarts}.
\end{proof}

\begin{theorem}[Infinum-Limit Value of $\gamma_{th}$]\label{Theorem:InfinumLimitSNRValueThatLowSNRRegimeStarts}
Any worst-fading environment leads in low-SNR regime when
\begin{eqnarray}\label{Eq:InfinumLimitSNRValueThatLowSNRRegimeStarts}
	\gammabar_{end}\lessapprox\gamma_{th}^{-}=0.64117587677~~(-1.93022825657~\text{\rm{dB}}).
\end{eqnarray}
\end{theorem}

\begin{proof}
Note that as the fading conditions get significantly worse, the diversity gain $m\in\mathbb{R}_{+}$ approaches to $m\rightarrow{1}/{2}$. In addition to this case, the fading conditions also gets worse when shape parameter (i.e., non-linearity) $\xi$ decreases (i.e., when $\xi\!\rightarrow\!{0}$). To the best of our knowledge, there is no physical justification in literature for such an extreme case that both the diversity order $m\rightarrow1/2$ and the shape parameter $\xi\rightarrow{0}$ simultaneously occur, specifically meaning that, for the worst-fading environments, the shape parameter $\xi$ cannot be smaller than one ($\xi>1$) and the diversity order $m$ cannot be smaller than half ($m\geq{1}/{2}$) as mentioned in \cite{BibAlouiniBook,BibNakagami1960}. Also explicitly mentioned in \cite{BibAlouiniBook}, the worst-case fading conditions are accurately characterized by a one-sided Gaussian distribution that follows the PDF
\begin{equation}\label{Eq:OneSidedGaussianPowerPDF}
	p_{\gamma_{end}}(\gamma;\gammabarsub{end})=\frac{\exp\left(-\frac{1}{2}{\gamma}/{\gammabarsub{end}}\right)}{\sqrt{2\pi\gammabarsub{end}\gamma}},
		\quad\text{for $\gamma\in[0,\infty)$},
\end{equation}
whose substitution into \eqref{Eq:HOCC} yields the HOCC $\widebar{C}(n;\gammabarsub{end})$ of one-sided Gaussian fading environments, that is 
\begin{equation}\label{Eq:HOECapacityForOneSidedGaussianFading}
\!\!\!\widebar{C}(n;\gammabarsub{end})=
	\fact{n}\,\FoxY[right]{n+1,n+2}{n+2,n+1}
			{2\gammabarsub{end}}
				{\left(1,1,0.5/\gammabar_{end},0.5\right),(1,1,0,1),\ldots,(1,1,0,1)}
					{(0,1,0,1),\ldots,(0,1,0,1)},\!\!\!
\end{equation}
where $\FoxYDefinition{m,n}{p,q}{\cdot}$ denotes the generalized version of Fox's H function, defined in \cite{BibYilmazAlouiniICT2010} by Yilmaz \& Alouini. Eventually, substituting \eqref{Eq:HOECapacityForOneSidedGaussianFading} into \eqref{Eq:BoundarySNRValueBetweenHighLowSNRRegimes}, we can find $\gamma_{th}^{-}$ (i.e., the infinum-limit value of $\gamma_{th}$). Although the resultant optimization has no analytical solution, the Pareto optimal solution can always be easily obtained by means of powerful optimization tools available in the mathematical software packages such as \Mathematica, \Maple and \Matlab. After accurately achieving optimization, the Pareto optimal $\gamma_{th}^{-}$ is obtained as \eqref{Eq:InfinumLimitSNRValueThatLowSNRRegimeStarts} and depicted in \figref{Fig:BoundaryTresholdSNRForLowSNRRegime}, which proves \theoremref{Theorem:InfinumLimitSNRValueThatLowSNRRegimeStarts}.
\end{proof}

As a consequence, with the aid of both \theoremref{Theorem:SupremumLimitSNRValueThatLowSNRRegimeStarts} and \theoremref{Theorem:InfinumLimitSNRValueThatLowSNRRegimeStarts}, we know that the SNR $\gamma_{th}$ at which the low-SNR regime starts will be $0.64117587677<\gamma_{th}<1.71828182846$ for any kind of fading environments. Consequently, using the results presented above, we can deal with both the HOCC and ACC in low-SNR regime of fading environments, especially by asymptotically tight upper-bounds as shown in the following.

\begin{theorem}\label{Theorem:HighOrderErgodicCapacityForLowSNRRegime}
An asymptotically tight upper-bound of $\widebar{C}(n;\gammabarsub{end})$ in low-SNR regime of fading environments is given by
\begin{equation}\label{Eq:HighOrderErgodicCapacityForLowSNRRegime}
	\widebar{C}(n;\gammabarsub{end})\lessapprox\hat{\mu}_{n}\,\gammabar_{end}^{n},\text{~~where~~}\hat{\mu}_{n}=AF_{n}(\gamma_{end})+1.
\end{equation}
\end{theorem}

\begin{proof}
According to \theoremref{Theorem:InfinumLimitSNRValueThatLowSNRRegimeStarts}, we know that $\gammasub{end}$ will follow the distribution most possibly taking random values much smaller than $\gamma_{th}^{-}=0.64117587677$. Using \cite[Eq. (1.511)]{BibGradshteynRyzhikBook} with this fact, we can show that $C\left(n;\gamma_{end}\right)$ is upper-bounded as $C\left(n;\gamma_{end}\right)\risingdotseq\gamma_{end}^{n}+O\left(\gamma_{end}^{n+1}\right)$. Accordingly, utilizing \eqref{Eq:HigherOrderAmountOfFading}, the HOCC $\widebar{C}(n;\gammabarsub{end})$ can be compactly approximated as in \eqref{Eq:HighOrderErgodicCapacityForLowSNRRegime}, which proves \theoremref{Theorem:HighOrderErgodicCapacityForLowSNRRegime}.
\end{proof} 

\begin{theorem}\label{Theorem:ErgodicCapacityForLowSNRRegime}
An asymptotically tight upper-bound of $\widebar{C}(\gammabarsub{end})$ in low-SNR regime of fading environments is given by
\begin{equation}\label{Eq:ErgodicCapacityForLowSNRRegime}
	\widebar{C}(\gammabarsub{end})\lessapprox\mathbb{E}[\gamma_{end}]=\gammabarsub{end}.
\end{equation}
\end{theorem}

\begin{proof}
The proof is obvious by setting $n=1$ in \eqref{Eq:HighOrderErgodicCapacityForLowSNRRegime}. 
\end{proof}

Note that the ACC $\widebar{C}(\gammabarsub{end})$ in low-SNR regime does not depend on the fading conditions except for $\gammabarsub{end}$. Since $\widebar{C}(\gammabarsub{end})$ cannot be greater than $\gammabarsub{end}$ in low-SNR regime, the main challenge is to improve the energy efficiency rather than increasing the diversity. In respect to the accuracy and analytical simplicity of \theoremref{Theorem:HighOrderErgodicCapacityForLowSNRRegime} and \theoremref{Theorem:ErgodicCapacityForLowSNRRegime}, let us consider some special cases for single link reception over generalized fading channels. For the fading conditions in a non-line-of-sight environment, the $\eta-\mu$ distribution is commonly used and it follows the PDF given by \cite[Eq. (18)]{BibYacoub2007}
\begin{equation}\label{Eq:EtaMuPowerPDF}
p_{\gamma_{end}}(\gamma)=\frac{2\sqrt{\pi}h^{\mu}}{\Gamma(\mu)}
	\biggl(\frac{\mu}{\gammabarsub{end}}\biggr)^{\mu+\frac{1}{2}}
		\biggl(\frac{\gamma}{H}\biggr)^{\mu-\frac{1}{2}}
			\exp\left(-\frac{2\mu{h}}{\gammabarsub{end}}\gamma\right)
				\BesselI[\mu-\frac{1}{2}]{\frac{2\mu{H}}{\gammabarsub{end}}\gamma},
\end{equation}
where $\mu$ represents the channel non-line-of-sight severity, and  $\eta$ denotes i) the ratio between the powers of the in-phase and quadrature components of the complex fading distribution(i.e., it is the first format by $0<\eta<\infty$), or ii) denotes the correlation between the inphase and quadrature components (i.e., it is the second format by $-1<\eta<1$). For these two formats, $h$ and $H$ are functions of $\eta$ and varies from one format to another \cite{BibYacoub2007} (i.e., $h=(2+\eta^{-1}+\eta)/4$ and $H=(\eta^{-1}-\eta)/4$ for the first format, whereas $h=1/(1-\eta^{2})$ and $H=\eta/(1-\eta^{2})$ for the second format). Note that, using the $n$th moment $\Expected{\gamma_{end}^{n}}$ \cite[Eq. (21)]{BibYacoub2007}, the HOCC $\widebar{C}(n;\gammabarsub{end})$ in low-SNR regime is obtained as
\begin{equation}\label{Eq:EtaMuHighOrderErgodicCapacityForLowSNRRegime}
\widebar{C}(n;\gammabarsub{end})\lessapprox
\frac{\Gamma(2\mu+n)}{h^{\mu+n}{(2\mu)}^{n}\Gamma(2\mu)}
	\Hypergeom{2}{1}{}{\mu+\frac{n}{2}+\frac{1}{2},\mu+\frac{n}{2}}{\mu+\frac{1}{2}}{{\left(\frac{H}{h}\right)}^2}
		\gammabar_{end}^{n},
\end{equation}
where substituting $n=1$ and then using \cite[Eq. (7.3.1/27)]{BibPrudnikovBookVol3} results in the ACC $\widebar{C}(\gammabarsub{end})\lessapprox\gammabarsub{end}$ in low-SNR regime.

The other example is the HOCC $\widebar{C}(n;\gammabarsub{end})$ of the most versatile fading environment, known as EGK fading distribution, whose PDF $p_{\gamma_{end}}(\gamma)$ is given in \eqref{Eq:EqInstantaneousSNRPDFForEGKFading} and its higher order AOF $AF_{n}(\gamma_{end})$ in \eqref{Eq:AmountOfFadingForEGKFading}. The HOCC $\widebar{C}(n;\gammabarsub{end})$ of the EGK fading environments is obtained as
\begin{equation}\label{Eq:EGKHighOrderErgodicCapacityForLowSNRRegime}
\!\!\!\widebar{C}(n;\gammabarsub{end})\lessapprox
	\frac{\Gamma\left(m+{n}/{\xi}\right)\Gamma\left(m_s+{n}/{\xi_s}\right)}
		{\Gamma(m)\Gamma(m_s)}
			\left(\beta\,\beta_s\right)^{-n}\gammabar_{end}^{n}
\end{equation}
for low-SNR regime. Furthermore, for the other fading environments, for instance using the higher order AOFs given in \eqref{Eq:AmountOfFadingForGeneralizedGammaFading}, \eqref{Eq:HigherOrderAmountOfFadingForLognormalFading} and \eqref{Eq:AmountOfFadingForKappaMuFading}, the HOCC in low-SNR region can be easily obtained.     

\section{Conclusion}\label{Sec:Conclusion}
In this article, we determine the boundary average SNRs between the high-SNR and low-SNR regimes in generalized fading environments. For these two regimes, we propose some novel closed-form expressions for the asymptotic analysis of both the HOCC and ACC in generalized fading environments, and show asymptotically tight and closed-form expressions, each of which is checked either by numerical numerical or simulation-based methods, for different fading environments commonly used in the literature. The methodological soundness of our asymptotic analysis clearly evidences that our closed-form expressions are insightful enough to comprehend the diversity gains, and their analytical accuracy and simplicity will be beneficial for the researchers, practitioners, and theoreticians in field of wireless communications.

\section*{Acknowledgment}
This work was supported by Y{\i}ld{\i}z Technical University (YTU).

 
\end{document}